\documentclass[11pt,letterpaper]{article}
\usepackage{yuzhao}

\newcommand{\odec}[1]{\breve{#1}}
\newcommand{\dec}[1]{\wt{#1}}
\newcommand{\Normal}{\mathrm{N}}    \newcommand{\normal}{\Normal}

\begin{document}

\title{Polynomial bounds for decoupling, with applications}
\author{Ryan O'Donnell\thanks{Computer Science Department, Carnegie Mellon University.  Supported in part by NSF grant CCF-1319743. \texttt{\{odonnell,yuzhao1\}@cs.cmu.edu}} \and Yu Zhao$^*$}
\maketitle

\begin{abstract}
    Let $f(x) = f(x_1, \dots, x_n) = \sum_{|S| \leq k} a_S \prod_{i \in S} x_i$ be an $n$-variate real multilinear polynomial of degree at most~$k$, where $S \subseteq [n] = \{1, 2, \dots, n\}$.  For its \emph{one-block decoupled} version,
    \[
        \odec{f}(y,z) = \sum_{\abs{S} \leq k} a_S \sum_{i \in S\vphantom{\ \}}} y_i \prod_{\mathclap{j \in S \setminus \{i\}}} z_j,
    \]
    we show tail-bound comparisons of the form
    \[
        \Pr\Brak{\abs{\odec{f}(\by,\bz)} > C_k t} \leq D_k \Pr\Brak{\abs{f(\bx)} > t}.
    \]
    Our constants $C_k, D_k$ are significantly better than those known for ``full decoupling''. For example, when $\bx, \by, \bz$ are independent Gaussians we obtain $C_k = D_k = O(k)$; when $\bx, \by, \bz$ are $\pm 1$ random variables we obtain $C_k = O(k^2)$,  $D_k = k^{O(k)}$.  By contrast, for full decoupling only $C_k = D_k = k^{O(k)}$ is known in these settings.

    We describe consequences of these results for query complexity (related to conjectures of Aaronson and Ambainis) and for analysis of Boolean functions (including an optimal sharpening of the DFKO~Inequality).
\end{abstract}
\thispagestyle{empty}

\newpage
\setcounter{page}{1}
\section{Introduction}

Broadly speaking, \emph{decoupling} refers to the idea of analyzing a complicated random sum involving dependent random variables by comparing it to a simpler random sum where some independence is introduced between the variables.  For perhaps the simplest example, if $(a_{ij})_{i,j=1}^n \in \R$ and $\bx_1, \dots, \bx_n, \by_1, \dots, \by_n$ are independent uniform $\pm 1$  random variables, we might ask how the moments of
\[
    \sum_{i,j=1}^n a_{ij} \bx_i \bx_j,  \text{ and its ``decoupled version''\ } \sum_{i,j=1}^n a_{ij} \bx_i \by_j
\]
compare.  The theory of decoupling inequalities developed originally in the study of Banach spaces, stochastic processes, and $U$-statistics, mainly between the mid-'80s and mid-'90s; see~\cite{dlPG99} for a book-length treatment.

The powerful tool of decoupling seems to be relatively under-used in theoretical computer science. (A recent work of Makarychev and Sidirenko~\cite{MS14} provides an exception, though they use a much different kind of decoupling than the one studied in this paper.)  In this work we will observe several places where decoupling can be used in a  ``black-box'' fashion to solve or simplify problems quite easily.

The main topic of the paper, however, is to study a partial form decoupling that we call ``one-block decoupling''.  The advantage of one-block decoupling is that for degree-$k$ polynomials we can achieve bounds with only \emph{polynomial} dependence on~$k$, as opposed to the exponential dependence on~$k$ that arises for the standard full decoupling.  Although one-block decoupling does not introduce as much independence as full decoupling does, we show several applications where one-block decoupling is sufficient.

The applications we describe in this paper are the following:
\begin{itemize}
    \item (Theorem~\ref{thm:our-aa}.) Aaronson and Ambainis's conjecture concerning the generality of their~\cite[Theorem~4]{AA15} holds.  I.e., there is a sublinear-query algorithm for estimating any bounded, constant-degree Boolean function.
    \item (Theorem~\ref{thm:AA-WLOG}.) The Aaronson--Ambainis Conjecture~\cite{Aar08,AA14} holds if and only if it holds for one-block decoupled functions. We also show how the best known result towards the conjecture can be proven extremely easily~\eqref{eqn:one-liner} in the case of one-block decoupled functions.
    \item (Corollary~\ref{cor:improved-dfko-tail}.) An optimal form of the DFKO Fourier Tail Bound~\cite{DFKO07}: any bounded Boolean function~$f$ that is far from being a junta satisfies $\sum_{\abs{S} > k} \wh{f}(S)^2 \geq \exp(-O(k^2))$.  Relatedly (Corollary~\ref{cor:improved-dfko-ineq}), any degree-$k$ real-valued Boolean function with $\Omega(1)$ variance and small influences must exceed~$1$ in absolute value with probability at least $\exp(-O(k^2))$; this can be further improved to $\exp(-O(k))$ if $f$ is homogeneous.
\end{itemize}

\subsection{Definitions}
Throughout this section, let $f$ denote a multilinear polynomial of degree at most~$k$ in $n$ variables $x = (x_1, \dots, x_n)$, with coefficients $a_S$ from a separable Banach space:
\[
    f(x) = \sum_{\substack{S \subseteq [n] \\ |S| \leq k}} a_S x_S,
\]
where we write $x_S = \prod_{i \in S} x_i$ for brevity.
(The coefficients $a_S$ will be real in all of our applications; however we allow them to be from a Banach space since the proofs are no more complicated.)

We begin by defining our notion of partial decoupling:
\begin{definition}
    The \emph{one-block decoupled} version of $f$, denoted $\odec{f}$, is the multilinear polynomial over $2n$ variables $y = (y_1, \dots, y_n)$ and $z = (z_1, \dots, z_n)$ defined by
    \[
        \odec{f}(y,z) = \sum_{\substack{S \subseteq [n] \\ 1 \leq |S| \leq k}} a_S \sum_{i \in S} y_i z_{S \setminus i}.
    \]
\end{definition}
In other words, each monomial term like $x_1x_3x_7$ is replaced with $y_1z_3z_7 + z_1y_3z_7 + z_1z_3y_7$.  In case $f$ is homogeneous we have the relation $\odec{f}(x,x) = k f(x)$.

Let us also recall the traditional notion of decoupling:
\begin{definition}
    The \emph{(fully) decoupled} version of $f$, which we denote by $\dec{f}$, is a multilinear polynomial over $k$ \emph{blocks} $x^{(1)}, \dots, x^{(k)}$ of $n$ variables; each $x^{(i)}$ is $x^{(i)} = (x^{(i)}_1, \dots, x^{(i)}_n)$. It is formed as follows:  for each monomial $x_S$ in~$f$, we replace it with the average over all ways of assigning its variables to different blocks.  More formally,
    \[
        \dec{f}(x^{(1)}, \dots, x^{(k)}) = a_\emptyset + \sum_{\substack{S \subseteq [n] \\ 1 \leq |S| \leq k}} \frac{(k-|S|)!}{k!}\cdot a_S \sum_{\substack{\text{injective} \\ b : S \to [k]}}\ \prod_{i \in S} x_i^{(b(i))}.
    \]
\end{definition}
The definition is again simpler if $f$ is homogeneous. For example, if $f$ is homogeneous of degree~$3$, then each monomial in~$f$ like $x_1x_3x_7$ is replaced in~$\dec{f}$ with
\[
    \frac16 	\paren{w_1y_2z_3+w_1z_2y_3+y_1w_2z_3+y_1z_2w_3+z_1w_2y_3+z_1y_2w_3}.
\]
(Here we wrote $w$, $y$, $z$ instead of $x^{(1)}$, $x^{(2)}$, $x^{(3)}$, for simplicity.)  Note that $\dec{f}(x, x, \dots, x) = f(x)$  always holds, even if $f$ is not homogeneous.\\

We conclude by comparing the two kinds of decoupling.  Assume for simplicity that $f$ is homogeneous of degree~$k$.  The fully decoupled version~$\dec{f}(x^{(1)}, \dots, x^{(k)})$ is in ``block-multilinear form''; i.e., each monomial contains exactly one variable from each of the~$k$ ``blocks''.  This kind of structure has often been recognized as useful in theoretical computer science; see, e.g.,~\cite{KN08,Lov10,KM13,AA15}. By contrast, the one-block decoupling~$\odec{f}(y,z)$ does not have such a simple structure; we only have that each monomial contains exactly one $y$-variable.  Nevertheless we will see several examples in this paper where having one-block decoupled form is just as useful as having fully decoupled form.  And as mentioned, we will show that it is possible to achieve one-block decoupling with only $\poly(k)$ parameter losses, whereas full decoupling in general suffers exponential losses in~$k$.
\begin{remark}  \label{rem:scaling}
    We have also chosen  different ``scalings'' for the two kinds of decoupling. For example, in the homogeneous case, we have $ \dec{f}(y, z, z, \dots, z) = \frac{1}{k} \cdot \odec{f}(y,z)$ and also $\Var[\dec{f}] = \frac{1}{(k-1)!} \Var[\odec{f}]$ for $f : \{\pm 1\}^n \to \R$.
\end{remark}

\subsection{A useful inequality}
Several times we will use the following basic inequality from analysis of Boolean functions, which relies on hypercontractivity; see~\cite[Theorems~9.24,~10.23]{OD14}.
\begin{theorem}                                     \label{thm:hypercon}
    Let $f(x) = \sum_{|S| \leq k} a_S x_S$ be a nonconstant $n$-variate multilinear polynomial of degree at most~$k$, where the coefficients $a_S$ are real.  Let $\bx_1, \dots, \bx_n$ be independent uniform $\pm 1$ random variables.  Then
    \[
        \Pr\bigl[f(\bx) > \E[f]\bigr] \geq \tfrac14 e^{-2k}.
    \]
    This also holds if some of the $\bx_i$'s are standard Gaussians.\footnote{Although it is not stated in~\cite{OD14}, an identical proof works since Gaussians have the same hypercontractivity properties as uniform $\pm 1$ random variables.}   Finally, if the $\bx_i$'s are not uniform $\pm 1$ random variables, but they take on each value $\pm 1$ with probability at least~$\lambda$, then we may replace $\tfrac14 e^{-2k}$ by $\tfrac14 (e^2/2\lambda)^{-k}$.
\end{theorem}

\section{Decoupling theorems, and query complexity applications}
\subsection{Classical decoupling inequalities, and an application in query complexity}
Traditional decoupling inequalities compare the probabilistic behavior of $f$ and $\dec{f}$ under independent random variables (usually symmetric ones; e.g., standard Gaussians).  The easier forms of the inequalities compare expectations under a convex test function; e.g., they can be used to compare $p$-norms.  The following was essentially proved in~\cite{dlPen92}; see~\cite[Theorem~3.1.1,(3.4.23)--(3.4.27)]{dlPG99}:
\begin{theorem}                                     \label{thm:classic-convex}
    Let $\Phi : \R^{\geq 0} \to \R^{\geq 0}$ be convex and nondecreasing.  Let $\bx = (\bx_1, \dots, \bx_n)$ consist of independent real random variables with all moments finite, and let $\bx^{(1)}, \dots, \bx^{(k)}$ denote independent copies. Then
    \[
        \E\Brak{\Phi\Paren{\norm{\dec{f}\paren{\bx^{(1)}, \dots, \bx^{(k)}}}}} \leq \E\Brak{\Phi\Paren{C_k\norm{f\paren{\bx}}}},
    \]
    where $C_k = k^{O(k)}$ is a constant depending only on~$k$.
\end{theorem}
\begin{remark}
    A reverse inequality also holds, with worse constant $k^{O(k^2)}$.
\end{remark}
Another line of research gave comparisons between tail bounds for $f$ and $\dec{f}$. This culminated in the following theorem from~\cite{dlPM95,Gin98}; see also~\cite[Theorem~3.4.6]{dlPG99}:
\begin{theorem}                                     \label{thm:classic-tail}
    In the setting of Theorem~\ref{thm:classic-convex}, for all $t > 0$,
    \[
        \Pr\Brak{\norm{\dec{f}\paren{\bx^{(1)}, \dots, \bx^{(k)}}} > C_kt} \leq D_k \Pr\Brak{\norm{f\paren{\bx}} > t},
    \]
    where $C_k = D_k = k^{O(k)}$. The analogous reverse bound also holds.
\end{theorem}
\begin{remark}  \label{rem:Kwapien}
    Kwapie\'{n}~\cite{Kwa87} showed that when the $\bx_i$'s are $\alpha$-stable random variables, the constant $C_k$ in Theorem~\ref{thm:classic-convex}, can be improved to $k^{k/\alpha}/k!$; this is $k^{k/2}/k!$ for standard Gaussians. Furthermore, for uniform $\pm 1$ random variables Kwapie\'{n}'s proof goes through as if they were $1$-stable; thus in this case one may take $C_k = k^k/k! \leq e^k$.  In the Gaussian setting with homogeneous~$f$, Kwapie\'{n} obtains $C_k = k^{k/2}/k!$ and $D_k = 2^k$ for Theorem~\ref{thm:classic-tail}.
\end{remark}
\begin{corollary}                                       \label{cor:classic-max}
    In the setting of Theorem~\ref{thm:classic-tail}, it holds that $\|\dec{f}\|_\infty \leq k^{O(k)} \|f\|_\infty$. Further, if $f \co \{\pm 1\}^n \to \R$ then $\|\dec{f}\|_\infty \leq (2e)^k \|f\|_\infty$.
\end{corollary}
\begin{proof}
    The first statement is an immediate corollary of either Theorem~\ref{thm:classic-convex} (taking $\Phi(u) = u^p$ and $p \to \infty$) or Theorem~\ref{thm:classic-tail} (taking $t = \|f\|_\infty$).  The second statement is immediate from Remark~\ref{rem:Kwapien}, with the better constant $k^k/k!$ in case $f$ is homogeneous.  In the general case, we use the fact that if $f^{=j}$ denotes the degree-$j$ part of~$f$, then $\|f^{=j}\|_\infty \leq 2^j \|f\|_\infty$; this is also proved by Kwapie\'{n}~\cite[Lemma~2]{Kwa87}.  Then
    \[
        \norm{\dec{f}}_\infty = \norm{\sum_{j=0}^k \dec{f^{=j}}}_\infty \leq \sum_{j=0}^k \norm{\dec{f^{=j}}}_\infty \leq \sum_{j=0}^k (j^j/j!) \norm{f^{=j}}_\infty \leq \sum_{j=0}^k (j^j/j!)2^j \norm{f}_\infty \leq (2e)^{k} \|f\|_\infty. \qedhere
    \]
\end{proof}
\begin{remark}
    Classical decoupling theory has not been too concerned with the dependence of constants on~$k$, and most statements like Theorem~\ref{thm:classic-tail} in the literature simply write $D_k = C_k$ to conserve symbols.  However there are good reasons to retain the distinction, since making $C_k$ small is usually much more important than making $D_k$ small.  For example, we can deduce Corollary~\ref{cor:classic-max} from Theorem~\ref{thm:classic-tail} regardless of $D_k$'s value.
\end{remark}

Let us give an example application of these fundamental decoupling results.  In a recent work comparing quantum query complexity to classical randomized query complexity, Aaronson and Ambainis~\cite{AA15} proved\footnote{Actually, there is a small gap in their proof. In the line reading ``By the concavity of the square root function\dots'', they claim that $\|\bX\|_1 \geq \|\bX\|_2$ when $\bX$ is a degree-$k$ polynomial of uniformly random $\pm 1$ bits.  In fact the inequality goes the other way in general. But the desired inequality does hold up to a factor of $e^k$ by~\cite[Theorem~9.22]{OD14}, and this is sufficient for their proof.} the following:
\begin{theorem}                                     \label{thm:aa1}
    Let $f$ be an $N$-variate degree-$k$ homogeneous block-multilinear polynomial with real coefficients.  Assume that under uniformly random $\pm 1$ inputs we have $\|f\|_\infty \leq 1$.  Then there is a randomized query algorithm making $2^{O(k)} (N/\eps^2)^{1-1/k}$ nonadaptive queries to the coordinates of $x \in \{\pm 1\}^N$ that outputs an approximation to $f(x)$ that is accurate to within $\pm \eps$ (with high probability).
\end{theorem}
The authors ``strongly conjecture[d]'' that the assumption of block-multilinearity could be removed, and gave a somewhat lengthy proof of this conjecture in the case of $k = 2$, using~\cite{DFKO07} .  We note that the full conjecture follows almost immediately from full decoupling:
\begin{theorem}                                     \label{thm:our-aa}
    Aaronson and Ambainis's Theorem~\ref{thm:aa1} holds without the assumption of block-multilinearity or homogeneity.
\end{theorem}
\begin{proof}
    Given a non-block-multilinear~$f$ on $N$ variables ranging in $\{\pm 1\}$, consider its full decoupling~$\dec{f}$ on $kN$ variables.  By Corollary~\ref{cor:classic-max} we have $\|\dec{f}\|_\infty \leq (2e)^k$.  Let $g = (2e)^{-k} \dec{f}$, so that $g : \{\pm 1\}^{kN} \to [-1,+1]$ is a degree-$k$ block-multilinear polynomial with $f(x) = (2e)^k g(x, x, \dots, x)$.  Now given query access to $x \in \{\pm 1\}^N$ and an error tolerance~$\eps$, we apply Theorem~\ref{thm:aa1} to $g(x, x, \dots, x)$ with error tolerance $\eps_1 = (2e)^{-k}\eps$; note that we can simulate queries to $(x, x, \dots, x)$ using queries to~$x$.  This gives the desired query algorithm, and it makes $2^{O(k)}(kN/\eps_1^2)^{1-1/k} = 2^{O(k)} (N/\eps^2)^{1-1/k}$ queries as claimed.  There is one more minor point: Theorem~\ref{thm:aa1} requires its function to be homogeneous in addition to block-multilinear.  However this assumption is easily removed by introducing $k$ dummy variables treated as $+1$, and padding the monomials with them.
\end{proof}

\subsection{Our one-block decoupling theorems, and the AA~Conjecture}           \label{sec:our-one-block}
We now state our new versions of Theorems~\ref{thm:classic-convex},~\ref{thm:classic-tail} which apply only to one-block decoupling, but that have \emph{polynomial} dependence of~$C_k$ on~$k$.  Proofs are deferred to Section~\ref{sec:proofs}.

As before, let $f(x) = \sum_{|S| \leq k} a_S x_S$ be an $n$-variate multivariate polynomial of degree at most~$k$ with coefficients $a_S$ in a Banach space; let $\bx = (\bx_1, \dots, \bx_n)$ consist of independent real random variables with all moments finite, and let $\by$,~$\bz$ be independent copies.  We  consider three slightly different hypotheses:
\begin{align*}
	\textbf{H1:} & \quad \text{$\bx_1, \dots, \bx_n \sim \Normal(0,1)$ are standard Gaussians.} \\
	\textbf{H2:} & \quad \text{$\bx_1, \dots, \bx_n$ are uniformly random $\pm 1$ values.} \\
	\textbf{H3:} & \quad \text{$\bx_1, \dots, \bx_n$ are uniformly random $\pm 1$ values and $f$ is homogeneous.}
\end{align*}

\begin{theorem} \label{thm:one-variable-full}
    If $\Phi : \R^{\geq 0} \to \R^{\geq 0}$ is convex and nondecreasing, then
    \[
        \E\Brak{\Phi\Paren{\norm{\odec{f}\paren{\by,\bz}}}} \leq \E\Brak{\Phi\Paren{C_k\norm{f\paren{\bx}}}}.
    \]
    Also, if $t > 0$ (and we assume $f$'s coefficients $a_S$ are real under \emph{\textbf{H2}, \textbf{H3}}), then
    \[
        \Pr\Brak{\norm{\odec{f}\paren{\by, \bz}} > C_kt} \leq D_k \Pr\Brak{\norm{f\paren{\bx}} > t}.
    \]
    Here
    \[
		C_k = \begin{cases}
		O(k) & \text{under \emph{\textbf{H1}}},\\
		O(k^{2}) &  \text{under \emph{\textbf{H2}}},\\
		O(k^{3/2}) &  \text{under \emph{\textbf{H3}}}, \qquad \qquad
		\end{cases}
		D_k = \begin{cases}
		O(k) & \text{under \emph{\textbf{H1}}},\\
		k^{O(k)} &  \text{under \emph{\textbf{H2, H3}}}.
		\end{cases}
	\]
\end{theorem}
\begin{remark}
    It may seem that for the $\Phi$-inequality in the Gaussian case, Kwapie\'{n}'s result mentioned in Remark~\ref{rem:Kwapien} is better than ours, since he achieves full decoupling with a better constant than we get for one-block decoupling.  But actually they are incomparable; the reason is the different scaling mentioned in Remark~\ref{rem:scaling}.
\end{remark}
\begin{remark}
    As we will explain later in Remark~\ref{rem:gauss-tight}, the bound $C_k = O(k)$ under \textbf{H1} is best possible (assuming that $D_k \leq \exp(O(k^2))$).
\end{remark}
An immediate consequence of the above theorem, as in Corollary~\ref{cor:classic-max}, is the following:
\begin{corollary}                                       \label{cor:one-vbl-max}
    If $f \co \{\pm 1\}^n \to \R$ then $\|\odec{f}\|_\infty \leq O(k^2) \|f\|_\infty$.
\end{corollary}

Let us now give an example of how one-block decoupling can be as useful as full decoupling, and why it is important to obtain $C_k = \poly(k)$.  A very notable open problem in analysis of Boolean functions is the \emph{Aaronson--Ambainis (AA) Conjecture}, originally proposed in 2008~\cite{Aar08,AA14}:
\paragraph{AA Conjecture.} \emph{Let $f : \{\pm 1\}^n \to [-1,+1]$ be computable by a multilinear polynomial of degree at most~$k$, $f(x) = \sum_{|S| \leq k} a_S x_S$.  Then $\MaxInf_i[f] \geq \poly(\Var[f]/k)$.}\\

Here we use the standard notations for influences and variance:
\[
    \MaxInf_i[f] = \max_{i \in [n]} \left\{\Inf_i[f]\right\}, \qquad \Inf_i[f] = \sum_{S \ni i} a_S^2, \qquad \Var[f] = \sum_{S \neq \emptyset} a_S^2, \qquad \|f\|_2^2 = \sum_{S} a_S^2.
\]

The AA~Conjecture is known to imply (and was directly motivated by) the following folklore conjecture concerning the limitations of quantum computation, dated to 1999 or before~\cite{AA14}:
\paragraph{Quantum Conjecture.} \emph{Any quantum query algorithm solving a Boolean decision problem using~$T$ queries can be correctly simulated on a $1-\eps$ fraction of all inputs by a classical query algorithm using $\poly(T/\eps)$ queries.}\\

Because of their importance for quantum computation, Aaronson has twice listed these conjectures as ``semi-grand challenges for quantum computing theory''~\cite{Aar05a,Aar10a}.\\

The best known result in the direction of the AA~Conjecture~\cite{AA14} obtains an influence lower bound of $\poly(\Var[f])/\exp(O(k))$, using the DFKO~Inequality~\cite{DFKO07}.  Here we observe that there is a ``one-line'' deduction of this bound under the assumption that $f$ is one-block decoupled.\footnote{This observation is joint with John Wright.}  To see this, suppose that $f$ is indeed one-block decoupled, so it can be written as
$f(y,z) = \sum_{i=1}^n y_i g_i(z)$, where  $g_i(z) = \sum_{S \ni i} a_S z_{S \setminus i}$ is the $i$th ``derivative'' of~$f$.
Observe that $\|g_i\|_2^2 = \Inf_i[f]$ and hence $\sum_{i=1}^n \|g_i\|_2^2 \geq \Var[f]$.  Also note that for any $z \in \{\pm 1\}^n$ we must have $\sum_{i=1}^n |g_i(z)| \leq 1$, as otherwise we could achieve $|f(y,z)| > 1$ by choosing $y \in \{\pm 1\}^n$ appropriately.  Taking expectations we get $\sum_{i=1}^n \|g_i\|_1 \leq 1$, and hence
\begin{multline}    \label{eqn:one-liner}
e^{k-1} \geq e^{k-1} \sum_{i=1}^n \|g_i\|_1 \geq \sum_{i=1}^n \|g_i\|_2 \geq \frac{\sum_{i=1}^n\|g_i\|_2^2}{\max_{i=1}^n \|g_i\|_2} \geq \frac{\Var[f]}{\max_{i=1}^n{\sqrt{\Inf_i[f]}}} \\ \Rightarrow\qquad \MaxInf[f] \geq e^{2-2k} \Var[f]^2, \qquad
\end{multline}
where the second inequality used the basic fact in analysis of Boolean functions~\cite[Theorem~9.22]{OD14} that $\|g\|_2 \leq e^{k-1}\|g\|_1$ for $g : \{\pm 1\}^n \to \R$ of degree at most $k-1$.

The above gives a good illustration of how even one-block decoupling can already greatly simplify arguments in analysis of Boolean functions.  We feel that~\eqref{eqn:one-liner} throws into sharp relief the challenge of improving  $\exp(-O(k))$ to $1/\poly(k)$ for the AA~Conjecture.  We now use our results to show that the assumption that $f$ is one-block decoupled is completely without loss of generality.
\begin{theorem}                                     \label{thm:AA-WLOG}
    The AA~Conjecture holds if and only if it holds for one-block decoupled functions~$f$.
\end{theorem}
\begin{proof}
    Suppose $f : \{\pm 1\}^n \to [-1,+1]$ has degree at most~$k$.  By Corollary~\ref{cor:one-vbl-max} we get that $\|\odec{f}\|_\infty \leq C_k = O(k^2)$.  Now $g = C_k^{-1} \odec{f}$ is one-block decoupled and has range $[-1,+1]$. Assuming the AA~Conjecture holds for it, we get some $i \in [2n]$ such that $\Inf_i[g] \geq \poly(\Var[g]/k)$.  Certainly this implies $\Inf_i[\odec{f}] \geq \poly(\Var[\odec{f}]/k)$.  Letting $i' = \max\{i,i-n\} \in [n]$, it's easy to see that $\Inf_{i'}[f] \geq \Inf_{i}[\odec{f}]/(k-1)$, and also that $\Var[\odec{f}] \geq \Var[f]$.  Thus $\Inf_{i'}[f] \geq \poly(\Var[f]/k)$, confirming the AA~Conjecture for~$f$.
\end{proof}
In particular, by combining this with~\eqref{eqn:one-liner} we recover the known $\poly(\Var[f])/\exp(O(k))$ lower bound for the AA~Conjecture as applied to general~$f$.
\begin{remark}
    Aaronson and Ambainis~\cite{AA15} recently observed that for the purposes of deriving the Quantum Conjecture, it suffices to prove the AA~Conjecture for fully decoupled~$f$.  However the AA~Conjecture is of significant interest in analysis of Boolean functions in and of itself, even independent of the Quantum Conjecture.  Thus we feel Theorem~\ref{thm:AA-WLOG} is worth knowing, especially in light of the simple argument~\eqref{eqn:one-liner}.
\end{remark}

\section{Tight versions of the DFKO theorems}
This section is concerned with analysis of Boolean functions $f : \{\pm 1\}^n \to \R$. We will use traditional Fourier notation, writing $f(x) = \sum_{S \subseteq [n]} \wh{f}(S) x_S$.  A key theme in this field is the dichotomy between functions with ``Gaussian-like'' behavior and functions that are essentially ``juntas''.  Recall that $f$ is said to be an $(\eps,C)$-junta if $\|f-g\|_2^2 \leq \eps$ for some $g : \{\pm 1\}^n \to \R$ depending on at most~$C$ input coordinates. Partially exemplifying this theme is a family of theorems stating that any Boolean function~$f$ which is not essentially a junta must have a large ``Fourier tail'' --- something like $\sum_{|S| > k} \wh{f}(S)^2 > \delta$. Examples of such results include Friedgut's Average Sensitivity Theorem~\cite{Fri98}, the FKN Theorem~\cite{FKN02} (sharpened in~\cite{JOW12,OD14}), the Kindler--Safra Theorem~\cite{KS02,Kin02}, and the Bourgain Fourier Tail Theorem~\cite{Bou02}.  The last of these implies that any $f : \{\pm 1\}^n \to \{\pm 1\}$ which is not a $(.01, k^{O(k)})$-junta must satisfy $\sum_{|S|> k} \wh{f}(S)^2 > k^{-1/2+o(1)}$.  This $k^{-1/2+o(1)}$ bound was made more explicit in~\cite{KN06}, and the optimal bound of $\Omega(k^{-1/2})$ was obtained in~\cite{KO12}.  These ``Fourier tail'' theorems have had application in fields such as PCPs and inapproximability~\cite{Kho02,Din07}, sharp threshold theory~\cite{FK96}, extremal combinatorics~\cite{EFF12}, and social choice~\cite{FKN02}.

All of the aforementioned theorems concern Boolean-\emph{valued} functions; i.e., those with range~$\{\pm 1\}$.  By contrast, the DFKO Fourier Tail Theorem~\cite{DFKO07} is a result of this flavor for \emph{bounded} functions; i.e., those with range~$[-1,+1]$.
\paragraph{DFKO Fourier Tail Theorem.} \emph{Suppose $f : \{\pm 1\}^n \to [-1,+1]$ is not an $(\eps, 2^{O(k)}/\eps^2)$-junta.  Then
\[
    \sum_{|S| > k} \wh{f}(S)^2 > \exp(-O(k^2 \log k)/\eps).
\]}

Most applications do not use this Fourier tail theorem directly. Rather, they use a key intermediate result, \cite[Theorem~3]{DFKO07}, which we will refer to as the ``DFKO Inequality''.  This was the case, for example, in a recent work on approximation algorithms for the Max-$k$XOR problem~\cite{BMO+15}.
\paragraph{DFKO Inequality.} \emph{Suppose $f : \{\pm 1\}^n \to \R$ has degree at most~$k$ and $\Var[f] \geq 1$.  Let $t \geq 1$ and suppose that $\MaxInf[f] \leq 2^{-O(k)}/t^2$.  Then
$
    \Pr[|f(\bx)| > t] \geq \exp(-O(t^2 k^2 \log k)).
$}\\

Returning to the theme of ``Gaussian-like behavior'' versus ``junta'' behavior, we may add that the DFKO results straightforwardly imply (by the Central Limit Theorem) analogous, simpler-to-state results concerning functions on Gaussian space and Hermite tails.  We record these generic consequences here; see, e.g.,~\cite[Sections~11.1,~11.2]{OD14} for a general discussion of such implications, and the definitions of Hermite coefficients~$\wh{f}(\alpha)$.
\begin{corollary}   \label{cor:dfko-Gaussian}
    Any $f : \R^n \to [-1,+1]$ satisfies the Hermite tail bound
    \[
        \sum_{|\alpha| > k} \wh{f}(\alpha)^2  > \exp(-O(k^2 \log k)/\Var[f]).
    \]
    Furthermore, suppose $\bz$ is a standard $n$-dimensional Gaussian random vector and $t \geq 1$.  Then any $n$-variate polynomial~$f$ of degree at most~$k$ with $\Var[f(\bz)] \geq 1$ satisfies $\Pr[|f(\bz)| > t] \geq \exp(-O(t^2 k^2 \log k))$.
\end{corollary}
Even though the Gaussian results in Corollary~\ref{cor:dfko-Gaussian} are formally easier than their Boolean counterparts, we are not aware of any way to prove them --- even in the case $n = 1$ --- except via DFKO.

\paragraph{Tightness of the bounds.} In~\cite[Section~6]{DFKO07} it is shown that the results in Corollary~\ref{cor:dfko-Gaussian} are tight, up to the $\log k$ factor in the exponent; this implies the same statement about the DFKO Fourier Tail Theorem and the DFKO Inequality.  The tight example in both cases is essentially the univariate, degree-$k$ Chebyshev polynomial.\footnote{Formally speaking, \cite[Section~6]{DFKO07} only argues tightness of the Boolean theorems, but their constructions are directly based on the degree-$k$ Chebyshev polynomial applied to a single standard Gaussian.} In the next section we will show how to use our one-block decoupling result to remove the $\log k$ in the exponential from both DFKO theorems.  The results immediately transfer to the Gaussian setting, and we therefore obtain the tight $\exp(-\Theta(k^2))$ bound for all versions of the inequality.

Our method of proof is actually to \emph{first} prove the results in the Gaussian setting, where the one-block decoupling makes the proofs quite easy.  Then we can transfer the results to the Boolean setting by using the Invariance Principle~\cite{MOO10}.  This methodology --- proving the more natural Gaussian tail bound first, then transferring the result to the Boolean setting via Invariance --- is quite reminiscent of how the optimal form of Bourgain's Fourier Tail Theorem was recently obtained~\cite{KO12}.

There is actually an additional, perhaps unexpected, bonus of our proof methodology; we show that the bound in the DFKO Inequality can be improved from $\exp(-O(t^2 k^2))$ to $\exp(-O(t^2 k))$ whenever $f$ is \emph{homogeneous}.

%

\subsection{Proofs of the tight DFKO theorems}
We begin with a tail-probability lower bound for one-block decoupled polynomials of Gaussians.
\begin{lemma} \label{lem:decoupledtail}
    Suppose $f(y,z) = \sum_{i=1}^n y_i g_i(z)$ is a one-block decoupled polynomial on $n+n$ variables, with real coefficients and degree at most~$k$.
    Let $\by, \bz \in \Normal(0,1)^n$ be independent standard $n$-dimensional Gaussians and write
    \begin{equation}    \label{eqn:variance-eq}
        \sigma^2 = \Var[f(\by,\bz)] = \sum_{i=1}^n \|g_i\|_2^2.
    \end{equation}
    Then for  $u > 0$ we have $\Pr[|f(\by,\bz)| > u] \geq \exp(-O(k + u^2/\sigma^2))$.
\end{lemma}
\begin{proof}
    Let $v(z) = \sum_{i=1}^n g_i(z)^2$, a polynomial of degree at most $2(k-1)$ in $z_1, \dots, z_n$.  By~\eqref{eqn:variance-eq} we have $\E[v(\bz)] = \sigma^2$.  We now use Theorem~\ref{thm:hypercon} to get
    \[
        \Pr[v(\bz) > \sigma^2] \geq \frac14 e^{-2(2k-1)} = \exp(-O(k)).
    \]
    On the other hand, for any outcome $\bz = z$ we have that $f(\by,z) \sim \Normal(0,v(z))$.  Thus
     \[
        v(z) > \sigma^2 \quad \implies\quad \Pr[|f(\by,z)| > \Omega(e^{-u^2/2\sigma^2}).
     \]
     Combining the previous two statements completes the proof, since $\by$ and $\bz$ are independent.
\end{proof}

We can now prove an optimal version of the DFKO Inequality in the Gaussian setting.  It is also significantly better in the homogeneous case.
\begin{theorem}     \label{thm:gaussian-dfko-ineq}
    Let $f : \R^n \to \R$ be a polynomial of degree at most~$k$, and let $\bx \sim \normal(0,1)^n$ be a standard $n$-dimensional Gaussian vector.  Assume $\Var[f(\bx)] \geq 1$.  Then for  $t \geq 1$ it holds that $\Pr[|f(\bx)| > t] \geq \exp(-O(t^2 k^2))$.  Furthermore, if $f$ is multilinear and homogeneous then the lower bound may be improved to $\exp(-O(t^2 k))$.
\end{theorem}
\begin{proof}
    It is well known that for any $n$-variate polynomial of Gaussians, we can find an $N$-variate multilinear polynomial of Gaussians of no higher degree that is arbitrarily close in L\'{e}vy distance (see, e.g.,~\cite[Lemma~15]{Kan11}, or use the CLT to pass to~$\pm 1$ random variables, then Invariance to pass back to Gaussians).  Note, however, that this transformation does not preserve homogeneity.  In any case, we can henceforth assume $f$ is multilinear, $f(x) = \sum_{|S| \leq k} a_S x_S$.

    For independent $\by, \bz \sim \normal(0,1)^n$, observe that
    \[
        \Var[\odec{f}(\by,\bz)] = \sum_{j=1}^k j \sum_{|S| = j} a_S^2 \geq \sum_{S \neq \emptyset} a_S^2 = \Var[f(\bx)] \geq 1,
    \]
    and if $f$ is homogeneous we get the better bound $\Var[\odec{f}(\by,\bz)] \geq k$.  By our Theorem~\ref{thm:one-variable-full} on one-block decoupling, we have
    \[
        \Pr\Brak{\Abs{f\paren{\bx}} > t} \geq D_k^{-1} \Pr\Brak{\Abs{\odec{f}\paren{\by, \bz}} > C_k t},
    \]
    where $C_k = D_k = O(k)$.  The theorem is now an immediate consequence of Lemma~\ref{lem:decoupledtail}.
\end{proof}
\begin{remark} \label{rem:gauss-tight}
    A consequence of this proof is that --- assuming $D_k \leq \exp(O(k^2))$ --- it is impossible to asymptotically improve on our  $C_k = O(k)$ in Theorem~\ref{thm:one-variable-full} in the Gaussian setting~\textbf{H1}. Otherwise, we would achieve a bound of $\exp(-o(k^2))$ in Theorem~\ref{thm:gaussian-dfko-ineq}, contrary to the example in~\cite[Section~6]{DFKO07}.
\end{remark}
We can now obtain the sharp DFKO Inequality in the Boolean setting by using the Invariance Principle.
\begin{corollary}                                       \label{cor:improved-dfko-ineq}
    Theorem~\ref{thm:gaussian-dfko-ineq} holds when $\bx \sim \{\pm 1\}^n$ is uniform and we additionally assume that $\MaxInf[f] \leq \exp(-C t^2 k^2)$, or just $\exp(-C t^2 k)$ in the homogeneous case.  Here $C$~is a universal constant.
\end{corollary}
\begin{proof}
    This follows immediately from the L\'{e}vy distance bound in~\cite[Theorem~3.19, Hypothesis~4]{MOO10}.  We only need to ensure that the L\'{e}vy distance is noticeably less than the target lower bound we're aiming for.  (We also remark that the Invariance Principle transformation preserves variance and homogeneity.)
\end{proof}
Next, we obtain the sharp DFKO Fourier Tail Theorem.  Its deduction from the DFKO Inequality in~\cite{DFKO07} is unfortunately not ``black-box'', so we will have to give a proof.
\begin{corollary}                                     \label{cor:improved-dfko-tail}
    Suppose $f : \{\pm 1\}^n \to [-1,+1]$ is not an $(\eps, 2^{O(k^2/\eps)})$-junta.  Then
    \begin{equation}            \label{eqn:tail-bound-internal}
        \sum_{|S| > k} \wh{f}(S)^2 > \exp(-O(k^2)/\eps).
    \end{equation}
\end{corollary}
\begin{proof}
    We use notation and basic results from~\cite{OD14}. Given $f : \{\pm 1\}^n \to [-1,+1]$, let $J = \{i \in [n] : \Inf^{\leq k}_i[f] > \exp(-Ak^2/\eps)\}$, where $A$ is a large constant to be chosen later.
    Since $\|f\|_2^2 \leq 1$ it follows easily that $|J| \leq 2^{O(k^2/\eps)}$.  Now define $g = f - f^{\subseteq J}$; note that $g$ has range in $[-2,+2]$ since $f^{\subseteq J}$ has range in $[-1,+1]$, being an average of~$f$ over the coordinates outside~$J$.  If $\|g\|^2_2 < \eps/2$ then $f$ is $\eps/2$-close to the $2^{O(k^2/\eps)}$-junta $f^{\subseteq J}$ and we are done.  Otherwise, $\|g\|^2_2 \geq \eps/2$ and we let $h = g^{\leq k}$.  If $\|h - g\|_2^2 > \eps/4$ then we immediately conclude that $\sum_{|S| > k} \wh{f}(S)^2 > \eps/4$, which is more than enough to be done.  Otherwise $\|h - g\|_2^2 \leq \eps/4$, from which we conclude $\|h\|_2^2 \geq \eps/4$.  Now~$h$ has degree at most~$k$ and satisfies $\Inf_i[h] \leq \exp(-Ak^2/\eps)$ for all $i \not \in J$.  Let $\wt{h}$ denote the mixed Boolean/Gaussian function which has the same multilinear form as~$h$, but where we think of the coordinates in~$J$ as being $\pm 1$ random variables and the coordinates not in~$J$ as being standard Gaussians.  We now ``partially'' apply the Invariance Principle~\cite[Theorem 3.19]{MOO10} to~$h$, in the sense that we only hybridize over the coordinates not in~$J$.  We conclude that the L\'{e}vy distance between $h$ and $\wt{h}$ is at most $\exp(-\Omega(A k^2 / \eps))$.  Our goal is now to show that 
    \begin{equation} \label{eqn:dfko-goal}
        \Pr[|\wt{h}| > 3] \geq \exp(-O(k^2/\eps)),
    \end{equation}
    where the constant in the $O(\cdot)$ does not depend on~$A$.  Having shown this, by taking $A$ large enough the L\'{e}vy distance bound lets us deduce~\eqref{eqn:dfko-goal} for~$h$ as well.  But then since $|g| \leq 2$ always, we may immediately deduce $\|g - h\|_2^2 \geq \exp(-O(k^2)/\eps)$ and hence~\eqref{eqn:tail-bound-internal}.

    It remains to verify~\eqref{eqn:dfko-goal}.  For each restriction $x_J$ to the $J$-coordinates, the function $\wt{h}_{x_J}$ is a multilinear polynomial in independent Gaussians with some variance~$\sigma^2_{x_J}$.   From Theorem~\ref{thm:gaussian-dfko-ineq} we can conclude that $\Pr[|\wt{h}_{x_J}| > 3] \geq \exp(-O(k^2/\sigma^2_{x_J}))$.  Thus if we can show $\sigma^2_{\bx_J} \geq \Omega(\eps)$ with probability at least $2^{-O(k)}$ when $\bx_J \in \{\pm 1\}^J$ is uniformly random, we will have established~\eqref{eqn:dfko-goal}.  But this follows similarly as in Lemma~\ref{lem:decoupledtail}.  Note that $\sigma^2_{x_J} = \E[\wt{h}_{x_J}^2]$, since $h$ has no constant term.  Now $\sigma^2_{x_J}$ is a degree-$2k$ polynomial in~$x_J$, and its expectation is simply $\|h\|_2^2 \geq \eps/4$, so Theorem~\ref{thm:hypercon} indeed implies that $\Pr[\sigma^2_{\bx_J} \geq \eps/4] \geq 2^{-O(k)}$ and we are done.
\end{proof}
\begin{remark}
    We comment that the dependence of $\MaxInf[f]$ on~$t$ in Corollary~\ref{cor:improved-dfko-ineq}, and the junta size in Corollary~\ref{cor:improved-dfko-tail}, are not as good as in~\cite{DFKO07} This seems to be a byproduct of the use of Invariance.
\end{remark}
A similar (but easier) proof can be used to derive the following Gaussian version of Corollary~\ref{cor:improved-dfko-tail}; alternatively, one can use a generic CLT argument, noting that the only ``junta'' a Gaussian function can be close to is a constant function:
\begin{corollary}   \label{cor:dfko-Gaussian-tail}
    Any $f : \R^n \to [-1,+1]$ satisfies the Hermite tail bound
    \[
        \sum_{|\alpha| > k} \wh{f}(\alpha)^2  > \exp(-O(k^2)/\Var[f]).
    \]
\end{corollary}
\noindent This strictly improves upon Corollary~\ref{cor:dfko-Gaussian}.

\ignore{	
\subsection{Tightness of the homogeneous inequality}
    In this section we consider the tightness of the $\exp(-O(t^2 k))$ bound in the homogeneous multilinear case of Theorem~\ref{thm:gaussian-dfko-ineq}.  We observe that neither the $\exp(-O(k))$ dependence nor the $\exp(-O(t^2))$ dependence can be improved.  Note that by CLT arguments, the same tightness holds in the Boolean homogeneous setting of Corollary~\ref{cor:improved-dfko-ineq}.
	\begin{proposition}
		\label{prop:tightk}
		For any $k$ and all sufficiently large~$n$, there is an $n$-variate degree-$k$ homogeneous multilinear polynomial $f: \R^n \to \R$ such that for $\bx \sim \Normal(0,1)^n$ a standard $n$-dimensional Gaussian we have $\Var[f(\bx)] \geq 1$ and
		\[
		  \Pr[|f(\bx)| > t] \leq \exp(-\Omega(kt^{2/k}))
		\]
		when $t \geq k^k/\sqrt{k!}$.
	\end{proposition}
	\begin{proof}
		Consider about degree-$k$ Hermite polynomial $H_k(x) = (-1)^ke^{x^2/2}\frac{d^k}{dx^k}e^{-x^2/2}$.
		The explicit expression shows that, if $x \geq k$,
		\[
		H_k(x) = \sum_{m = 0}^{\lfloor k/2 \rfloor}(-1)^m \frac{k!}{m!(k-2m)!} \frac{x^{k-2m}}{2^m} \leq x^k.
		\]
		Therefore
		\[
		\Pr_{\bX_0} [H_k(\bX_0) \geq t] \leq \Pr_{\bX_0} [\bX_0^k \geq t] \leq \exp(-\Omega(t^{2/k})).
		\]
		when $t \geq k^k$.
		It is well-known (see\cite[Eq.~(11.9)]{OD14}) that with Gaussian input $\bX_0 \sim N(0,1)$,
		\[
		\E_{\bX_0}[H_k(\bX_0)] = 0, \qquad \Var_{\bX_0}[H_k(\bX_0)] = k!.
		\]
		Define the normalized Hermite polynomial as $h_k = \frac1{\sqrt{k!}} H_k$. We get the tail bound
		\[
		\Pr_{\bX_0} [h_k(\bX_0) \geq t] \leq \exp(-\Omega((t\sqrt{k!})^{2/k})) = \exp(-\Omega(kt^{2/k}))
		\]
		when $t \geq k^k/\sqrt{k!}$.
		If we define $\bX_0 = \frac{1}{\sqrt{N}}(\bX_1 + \dots + \bX_N)$ where $\bX_1, \dots, \bX_N \sim N(0, 1)$ are independent standard Gaussian random variables, and
		\[
		f_{k, N}(\bX_1, \dots, \bX_N) = N^{-k/2} \sum_{1 \leq j_1 < \dots, <j_k \leq N} \bX_{j_1}\cdots \bX_{j_k}.
		\]
		Hermite polynomial $h_k$ can be approximated by homogeneous multilinear polynomial $f_{k, N}$ with same degree $k$:
		\[
		h_k(\bX_0) = \lim_{N \to \infty} f_{k, N}(\bX_1, \dots, \bX_N)
		\]
		where the convergence holds in any $L^p$ (see \cite[Eq.~(38)]{adamczak2014}). Therefore
		\[
		\Pr[f_{k, N}(\bX_1, \dots, \bX_N) \geq t] \leq \exp(-\Omega(kt^{2/k}))
		\]
		for some large enough $N$.
		Notice that $\Var[f_{k, N}]$ is not exactly 1, but for large enough $N$, $.99 \leq \Var[f_{k, N}] \leq 1$, so the tail bound still holds after the function is normalized.
	\end{proof}
	
	\begin{proposition}
		\label{prop:tightk2}
		For any constant $k$, for any $t > 1$, there exist some degree-$k$ homogeneous Gaussian function $f: \R^n \to \R$ with $\Var[f] = 1$ such that
		\[
		\Pr_{\bX \sim N(0,1)^n}[|f(\bX)| \geq t] \leq \exp(-\Omega(t^2))
		\]
	\end{proposition}
	\begin{proof}
		Consider about function
		\[
		f_n(x) = \frac{1}{\sqrt{n}}(x_1 x_2 \dots x_k+ x_{k+1}x_{k+2}\dots x_{2k}+\dots + x_{(n-1)k+1}x_{(n-1)k+2} \dots x_{nk})
		\]
		We know that if $\bX_1, \dots, \bX_k \sim N(0,1)$ are independent Gaussians,  $\E[\bX_1\bX_2\dots\bX_k] = 0, \E[(\bX_1\bX_2\dots\bX_k)^2] = 1$ and $\E[|\bX_1\bX_2\dots\bX_k|^3]= \rho^k$ where $\rho = \E_{\bX \sim N(0,1)}[|\bX|^3]$ is some constant not related to $k$.
		By Berry-Esseen Theorem,
		\[
		\left|\Pr_{\bX \sim N(0,1)^{kn}}[|f_n(\bX)| \geq t] - \Pr_{\bX_0 \sim N(0,1)}[|\bX_0| \geq t]
		\right| \leq \frac{C\rho^k}{\sqrt{n}}.
		\]
		If we choose a sufficient large $n \geq \rho^{2k} \exp(t^6) $, function $f_n$ will approximate standard Gaussian distribution. Therefore,
		\[
		\Pr_{\bX \sim N(0,1)^{kn}}[|f_n(\bX)| \geq t] \leq \Pr_{\bX_0 \sim N(0,1)}[|\bX_0| \geq t] + \frac{C\rho^k}{\sqrt{n}} \leq   \exp(-t^2) - C\exp(-t^3) \leq \exp(-\Omega(t^2)) \qedhere
		\]
	\end{proof}
}
		
\section{Proofs of our one-block decoupling theorems}           \label{sec:proofs}
In this section we prove Theorem~\ref{thm:one-variable-full}.  The key idea of the proof is to express $\odec{f}(y, z)$ as a ``small'' linear combination of expressions of the form $f(\alpha_i x + \beta_i y)$, where $\alpha_i^2 + \beta_i^2 = 1$ (in the Gaussian case) or $|\alpha_i| + |\beta_i| = 1$ (in the Boolean case).  The following is the central lemma.

\begin{lemma}		\label{lem:summation}
    In the setting of Theorem~\ref{thm:one-variable-full}, there exists $m = O(k)$ and   $\alpha, \beta, c \in \R^{m}$ such that
    \begin{itemize}
        \item $\odec{f}(y, z) = \sum_{i=1}^{m} c_if(\alpha_iy + \beta_iz)$;
	    \item $\|c\|_1 \leq C_k$;
        \item $\alpha_i^2 + \beta_i^2 = 1$ for all $i \in [m]$ under \emph{\textbf{H1}}, and $|\alpha_i| + |\beta_i| = 1$ for all $i \in [m]$ under \emph{\textbf{H2}, \textbf{H3}};
	   \item $|\alpha_i|, |\beta_i| \geq 1/O(C_k)$ for all $i \in [m]$.
    \end{itemize}
\end{lemma}
	
With Lemma~\ref{lem:summation} in hand, the proof of Theorem~\ref{thm:one-variable-full} is quite straightforward in the Gaussian case, and not much more difficult in the Boolean case. We show these deductions first.

\begin{proof}[Proof of Theorem~\ref{thm:one-variable-full} under Hypothesis \emph{\textbf{H1}}]
	By Lemma~\ref{lem:summation}, for any convex nondecreasing function $\Phi : \R^{\geq 0} \to \R^{\geq 0}$ we have
	\begin{align*}
		\E\Brak{\Phi\Paren{\norm{\odec{f}\paren{\by,\bz}}}}  &= \E\Brak{\Phi\Paren{\Norm{\sum_{i=1}^m c_i f\paren{\alpha_i \by +\beta_i \bz}}}} \\
		&\leq \E\Brak{\Phi\Paren{\sum_{i=1}^m |c_i| \Norm{f\paren{\alpha_i \by +\beta_i \bz}}}} \\
		&\leq \sum_{i = 1}^{m}\frac{|c_i|_{\phantom{1}}}{\|c\|_1} \E\left[\Phi\left(\|c\|_1 \|f(\alpha_i \by + \beta_i \bz)\|\right)\right]\\
		&= \sum_{i=1}^{m}\frac{|c_i|_{\phantom{1}}}{\|c\|_1} \E[\Phi(\|c\|_1 \|f(\bx)\|)]\\
		&\leq \E[\Phi(C_k\|f(\bx)\|].
	\end{align*}
    Here the inequalities follow from the convexity and monotonicity of $\Phi$, and the second equality holds because $\alpha_i \by + \beta \bz \sim \Normal(0,1)^n$ due to $\alpha_i^2 + \beta_i^2 = 1$.
	
    As for the tail-bound comparison, by Lemma~\ref{lem:summation}, whenever $y, z$ are such that $\|\odec{f}(y, z)\| > C_k t$, the triangle inequality implies that there must exist at least one $i \in [m]$ with  ${\|f(\alpha_i y + \beta_i z)\| > t}$. It follows that there must exist at least one $i \in [m]$ such that
    \[
        \Pr[\|f(\alpha_i \by + \beta_i \bz)\| > t] \geq \frac1m \Pr[\|\odec{f}(\by, \bz)\| > C_k t].
    \]
    This completes the proof, since $\alpha_i \by + \beta_i \bz \sim \Normal(0,1)^n$ and $m = O(k)$.
\end{proof}
	
\begin{proof}[Proof of Theorem~\ref{thm:one-variable-full} under Hypotheses \emph{\textbf{H2}, \textbf{H3}}]
	We define $\pm 1$ random variables as follows:
	\[
    	\bx^{(i)}_j = \begin{cases}
    	\operatorname{sgn}(\alpha_i) \by_j & \text{with probability } |\alpha_i|,\\
    	\operatorname{sgn}(\beta_i) \bz_j & \text{with probability } |\beta_i|,
    	\end{cases}
   	\]
    for all $i \in [m]$ and $j \in [n]$ independently. Notice that each $\bx^{(i)}$ is distributed uniformly on $\{\pm 1\}^n$, though they are not independent. To prove the desired inequality concerning~$\Phi$, we can repeat the proof in the Gaussian case, except that we no longer have the identity
    \[
        \E\left[\Phi\left(\|c\|_1 \|f(\alpha_i \by + \beta_i \bz)\|\right)\right]
		= \E[\Phi(\|c\|_1 |f(\bx)|)].
    \]
    In fact we will show that the left-hand side is at most the right-hand side.  Notice that for all fixed $y,z \in \{\pm 1\}^n$, the multilinearity of $f$ implies that
    \begin{equation}                \label{eqn:its-multilinear}
        f(\alpha_i y + \beta_i z) = \E[f(\bx^{(i)}) \mid (\by,\bz) = (y,z)].
    \end{equation}
    Thus
    \begin{align*}
        \E\left[\Phi\left(\|c\|_1 \|f(\alpha_i \by + \beta_i \bz)|\right)\right] &= \E_{\by,\bz}\left[\Phi\left(\|c\|_1 \left\| \E_{\bx^{(i)} \mid \by,\bz} \left[f(\bx^{(i)})\right]\right\|\right)\right]\\ &\leq \E_{\by,\bz}\E_{\bx^{(i)}}\left[\Phi\left(\|c\|_1 \|f(\bx^{(i)})\|\right)\right] = \E\left[\Phi\left(\|c\|_1 \|f(\bx)\|\right)\right],
    \end{align*}
    as claimed, where we used convexity again.

    As for the tail-bound comparison, recall that we are now assuming $f$ has real coefficients.  As in the Gaussian case there is at least one  $i \in [m]$ with
    \[
        \Pr[|f(\alpha_i \by + \beta_i \bz)| > t] \geq \frac{1}{O(k)} \Pr[|\odec{f}(\by, \bz)| > C_k t].
    \]
    Now suppose $y, z$ are such that $|f(\alpha_i y + \beta_i z)| > t$ and consider the conditional distribution on~$\bx^{(i)}$.  If we can show that, conditionally, $\Pr[|f(\bx^{(i)})| > t] \geq k^{-O(k)}$ then we are done.  But from~\eqref{eqn:its-multilinear} we have that $\abs{\E[f(\bx^{(i)})]} > t$; therefore the desired result follows from Theorem~\ref{thm:hypercon} and the fact that $\min(|\alpha_i|, |\beta_i|) \geq 1/O(C_k) = 1/\poly(k)$.
\end{proof}

\subsection{Proof of Lemma~\ref{lem:summation}}

The proof of Lemma~\ref{lem:summation} involves minimizing $\|c\|_1$ by carefully setting the ratios of $\alpha_i$ and $\beta_i$ to be a hyperharmonic progression.
	
\begin{proof}[Proof of Lemma~\ref{lem:summation}]
    The main work involves treating the homogeneous case.
    \paragraph{Homogeneous case.}  Our goal for homogeneous $f$ is to write
	\[
	   \odec{f}(y, z) = \sum_{i=1}^{k+1} c_if(\alpha_iy + \beta_iz).
	\]
	Comparing the expressions term by term, it is equivalent to say that for any $S \subseteq [n]$ with $|S| = k$,
	\[
	   \sum_{j \in S} y_j z_{S/j} = \sum_{i = 1}^{k + 1} c_i \prod_{j \in S} (\alpha_i y_j + \beta_i z_j).
	\]
    We can further simplify this to the conditions
	\begin{equation}   \label{eqn:conditions}
	\sum_{i = 1}^{k + 1} c_i \alpha_i^{k-t} \beta_i^t = \begin{cases}
	   1 & \text{ if } t = k - 1\\
	   0 & \text{ otherwise}
	   \end{cases}
	\end{equation}
	for all integers $ 0 \leq t \leq k$.	Let us write $\Delta_i = \frac{\beta_i}{\alpha_i}$ and introduce the Vandermonde matrix
	\[
	V =
	\begin{bmatrix}
	1 & 1 & \dots & 1\\
	\Delta_1 & \Delta_2 &  \cdots& \Delta_{k+1}\\
	\cdots & \cdots & \cdots & \cdots\\
	\Delta_1^{k-1} & \Delta_2^{k-1} &  \cdots & \Delta_{k+1}^{k-1}\\
	\Delta_1^k & \Delta_2^k &  \cdots& \Delta_{k+1}^k
	\end{bmatrix}.
	\]
    We will also write $A$ for the diagonal matrix $\text{diag}(\alpha_1^k, \alpha_2^k, \dots, \alpha_{k+1}^k)$, and write $e_k$ for the indicator vector of the $k$th coordinate, $e_k = (0, 0, \dots, 0, 1, 0)$. Then the necessary conditions~\eqref{eqn:conditions} are equivalent to the matrix equation $VAc = e_k$.  Assuming all the $\Delta_i$'s are different, $V$ is invertible and there is an explicit formula for its inverse~\cite{MS58}.  This yields the following expression for the $c_1, \dots, c_{k+1}$ in terms of $\alpha$ and $\beta$:
	\begin{equation}   \label{eqn:c-def}
        c_i = (A^{-1}V^{-1}e_{k})_i =  \frac{1}{\alpha_i^k} \cdot \frac{\Delta_i - \sum_{j=1}^{k+1} \Delta_j}{\prod_{j=1, j \neq i}^{k+1} (\Delta_i - \Delta_j) }.
	\end{equation}
	
    \paragraph{The main illustrative case: Hypothesis H1 and $k$ odd.} We will now assume that $k$ is odd; this assumption will be easily removed later.  It will henceforth be convenient to replace our indices $1, \dots, k+1$ with the following slightly peculiar set of indices:
    \[
        I = \bigl\{\pm 1, \pm 2, \dots, \pm \tfrac{k-1}{2}, \pm \tfrac12\bigr\}.
    \]
    Now under Hypothesis~\textbf{H1}, we will choose
	\[
        \alpha_i = \frac{i}{\sqrt{k^2 + i^2}}, \quad \beta_i = \frac{k}{\sqrt{k^2 + i^2}} \quad\implies\quad \Delta_i = \frac{k}{i}
	\]
    for all $i \in I$.  These choices satisfy $\alpha_i^2 + \beta_i^2 = 1$ and $|\alpha_i|, |\beta_i| \geq 1/O(C_k)$, so it remains to prove that for~$c$ defined by~\eqref{eqn:c-def} we have $\|c\|_1 \leq O(k)$.
	
    Let us upper-bound all $|c_i|$. Since it easy to see that $|c_i| = |c_{-i}|$ for all $i \in I$, it will suffice for us to consider the positive $i \in I$.  For $1 \leq i \leq \frac{k-1}2$, we have
	\begin{align*}
		\left|\prod_{j \in I, j \neq i} (\Delta_i - \Delta_j) \right|
		&= (\Delta_{1/2} - \Delta_i) (\Delta_i-\Delta_{-1/2}) \cdot \prod_{j=1, j \neq i}^{(k-1)/2} \left|\Delta_i - \Delta_j \right| \cdot \prod_{j=-(k-1)/2}^{-1} (\Delta_i - \Delta_j) \\
		&= \left(2k - \frac{k}i\right)\left(2k + \frac{k}i\right) \cdot \prod_{j=1, j \neq i}^{(k-1)/2} \left|\frac{k}{i} - \frac{k}{j} \right| \cdot \prod_{j=1}^{(k-1)/2} \left(\frac{k}i + \frac{k}j\right) \\
		&= k^{k}\left(4-\frac1{i^2}\right) \cdot \prod_{j=1, j \neq i}^{(k-1)/2} \frac{|j-i|}{ij} \cdot \prod_{j=1}^{(k-1)/2}\frac{j+i}{ij}\\
		&= \frac{k^{k}}{i^{k-2}} \left(4-\frac1{i^2}\right) \frac{\left(\frac{k-1}2 + i\right)! \left(\frac{k-1}2 - i\right)!}{ \left(\frac{k-1}2 \right)!^2}.
	\end{align*}
	Thus from~\eqref{eqn:c-def},
	\begin{align*}
		|c_i|
		&= \left(\frac{\sqrt{k^2+i^2}}{i}\right)^k \cdot \frac{k}{i} \cdot \frac{i^{k-2}}{k^k} \cdot  \frac1{4-1/i^2} \cdot \frac{ \left(\frac{k-1}2 \right)!^2}{\left(\frac{k-1}2 + i\right)! \left(\frac{k-1}2 - i\right)!}\\
		&= \frac{k}{i^3} \left(1 + \frac{i^2}{k^{2}}\right)^{k/2} \frac1{4-1/i^2} \frac{ \left(\frac{k-1}2 \right)!^2}{\left(\frac{k-1}2 + i\right)! \left(\frac{k-1}2 - i\right)!}.
	\end{align*}
	When $ 1 \leq i \leq \sqrt{k}$, we have
	\begin{align*}
		|c_i| &= \frac{k}{i^3} \left(1 + \frac{i^2}{k^{2}}\right)^{k/2} \frac1{4-1/i^2}  \frac{ \left(\frac{k-1}2 \right)!^2}{\left(\frac{k-1}2 + i\right)! \left(\frac{k-1}2 - i\right)!} \leq \frac{k}{i^3} \left(1 + \frac{1}{k}\right)^{k/2}
		\leq \frac{\sqrt{e}k}{i^3}.
	\end{align*}
	For $\sqrt{k} \leq i \leq \frac{k-1}2$, consider the ratio between $(i+1)^3|c_{i+1}|$ and $i^3 |c_i|$; it satisfies
	\begin{align*}
		\frac{(i+1)^3|c_{i+1}|}{i^3|c_i|} &\leq \frac{(k^{2} + (i + 1)^2)^{k/2}}{(k^{2}+i^2)^{k/2}} \cdot \frac{\frac{k-1}2 - i}{\frac{k-1}2 + i +1} \\
		&=\left(1+\frac{2i+1}{k^2+i^2}\right)^{k/2}\cdot \frac{k-1-2i}{k+1+2i}\\
		&\leq \left(1+\frac{2i+1}{k^2}\right)^{k/2} \cdot \frac{k-1-2i}k\\
		&\leq e^{\frac{2i+1}{2k}}\left(1 - \frac{2i+1}k\right) \leq 1.
	\end{align*}
	The last inequality holds since $e^{x/2}(1-x) \leq 1$ for all $0 \leq x \leq 1$. Thus we have $(i+1)^3|c_{i+1}| \leq i^3|c_i|$, and hence by induction that
	\begin{equation}   \label{eqn:c-bound1}
	|c_i| \leq \frac{\sqrt{e}k}{i^3}   \quad \forall\ 1 \leq i \leq \tfrac{k-1}{2}.
	\end{equation}
	
	We now need to bound $c_{1/2}$.  Similarly to the above, we have
	\begin{align*}
		\left|\prod_{j \in I, j \neq \frac12} (\Delta_{1/2} - \Delta_j) \right|
		&= (\Delta_{\frac12} - \Delta_{-1/2}) \cdot \prod_{j=1}^{(k-1)/2} (\Delta_{1/2} - \Delta_j) \cdot \prod_{j=-(k-1)/2}^{-1} (\Delta_{\frac12} - \Delta_j) \\
		&= 4k  \cdot \prod_{j=1}^{(k-1)/2}  \left(2k - \frac{k}j \right) \cdot\prod_{j=1}^{(k-1)/2} (2k + \frac{k}j) \\
		&= 4k^{k} \cdot \prod_{j=1}^{(k-1)/2} \frac{2j-1}{j} \cdot \prod_{j=1}^{(k-1)/2}\frac{2j+1}{j}\\
		&= 4k^{k}  \frac{(k-2)!!k!!}{ \left(\frac{k-1}2 \right)!^2}
	\end{align*}
	Thus from~\eqref{eqn:c-def} we get
	\begin{align}
		|c_{1/2}| &= \frac{(\sqrt{k^2+(1/2)^2})^k}{(1/2)^k} \cdot 2k \cdot \frac{1}{4k^{k}} \cdot  \frac{ \left(\frac{k-1}2 \right)!^2}{(k-2)!!k!!}\nonumber\\
		& = \left(1 + \frac{1}{4k^2}\right)^{k/2} \left(\frac{(k-1)!!}{(k-2)!!}\right)^2 \leq  4k. \label{eqn:c-bound2}
	\end{align}
    Now combining~\eqref{eqn:c-bound1},~\eqref{eqn:c-bound2}, we obtain
	\[
	   \|c\|_1 = 2 \sum_{i = 1}^{(k-1)/2}|c_i| + 2|c_{1/2}|\leq 2\sqrt{e}\sum_{i =1}^{(k-1)/2} \frac{k}{i^3} + 8k \leq 20k,
	\]
    as needed.
	
    \paragraph{Handling even $k$.} 	If $k$ is even, we define our index set to be
    \[
        I = \bigl\{0, \pm 1, \pm 2, \dots, \pm \tfrac{k-2}{2}, \pm \tfrac12\bigr\}.
    \]
    For $i \in I \setminus \{0\}$ we define $\alpha_i$ and $\beta_i$ as before; we also define $\alpha_0 = 1$, $\beta_0 = 0$, and hence $\Delta_0 = 0$. It is easy to check that $c_0 = 0$ (and hence we haven't actually violated $|\beta_i| \geq 1/O(C_k)$), and the upper bounds for the other $|c_i|$ still hold.  This completes the proof of the homogeneous case under Hypothesis~\textbf{H1}.
	
    \paragraph{Hypotheses H3.} We explain the case of $k$ odd; the same trick as before can be used for even~$k$.  For Hypothesis~\textbf{H3} we use
	\[
	   \alpha_i = \frac{i}{k^{3/2}+|i|}, \quad \beta_i = \frac{k^{3/2}}{k^{3/2} + |i|} \quad\implies\quad \Delta_i = \frac{k^{3/2}}{i},
	\]
    which satisfy $|\alpha_i| + |\beta_i| = 1$ and $|\alpha_i|, |\beta_i| \geq 1/O(k^{3/2})$.  Analysis similar to before shows that $\|c\|_1 \leq O(k^{3/2})$.  This completely finishes the proof under Hypothesis~\textbf{H3}.

    \paragraph{Hypothesis H2, the homogeneous case.} Here we do something slightly different.  For even or odd~$k$ we let the index set be $I = \{1, 2, \dots, k, \frac12\}$ and then define
	\[
	   \alpha_i = \frac{i^2}{k^{2}+i^2}, \quad \beta_i = \frac{k^{2}}{k^{2} + i^2} \quad\implies\quad \Delta_i = \frac{k^{2}}{i^2}.
	\]
    Now we have $|\alpha_i| + |\beta_i| = \alpha_i + \beta_i = 1$ and $|\alpha_i|, |\beta_i| \geq 1/O(k^2)$.  Again, similar analysis shows that $\|c\|_1 \leq O(k^{2})$.
	
    \paragraph{Extending to the non-homogeneous case under H2.} Now we need to be concerned with the terms at degree $k' < k$.  Here a key observation is that, since $\alpha_i + \beta_i = 1$ for all $i$, the following holds for all $k' < k$:
	\[
	   \sum_{i} c_i \alpha_i^{k' - t} \beta_i^{t} = \sum_{i} c_i \alpha_i^{k' - t} \beta_i^t (\alpha_i + \beta_i) = \sum_i c_i \alpha_i^{k'-t+1}\beta_i^t + \sum_i c_i \alpha_i^{k'-t}\beta_i^{t+1}.
	\]
	Thus an induction shows that in fact
	\[
	\sum_{i} c_i \alpha_i^{k'-t} \beta_i^t = \begin{cases}
	k-k' & \text{ if } t = k'\\
	1 & \text{ if } t = k' - 1\\
	0 & \text{ otherwise}
	\end{cases}
	\]
	for all $k' \leq k$. This is almost exactly what we need to treat the non-homogeneous case using all the same choices for $c, \alpha, \beta$, except for the $t = k'$ case.  But we can use a simple trick to fix this:
	\[
	\frac12\sum_{i} c_i \alpha_i^{k'-t} \beta_i^t - \frac12\sum_{i} c_i (-\alpha_i)^{k'-t} \beta_i^t = \frac{1 - (-1)^{k'-t}}2 \sum_{i} c_i \alpha_i^{k'-t} \beta_i^t = \begin{cases}
	1 & \text{ if } t = k' - 1\\
	0 & \text{ otherwise}
	\end{cases}
	\]
	From this we get
	\[
	\odec{f}(y, z) = \sum_{i = 1}^{m}c_i f(\alpha_i y + \beta_i z)
	\]
	even in the non-homogeneous case, with all the desired conditions and $m = 2(k+1)$.
	
	\paragraph{Extending to the non-homogeneous case under H1.} The trick here for handling degree $k' < k$ is similar.  Using the fact that $\alpha_i^2 + \beta_i^2 = 1$ for all $i$, we get that for all $k' < k$,
	\[
	\sum_{i} c_i \alpha_i^{k' - t} \beta_i^{t} = \sum_{i} c_i \alpha_i^{k' - t} \beta_i^t (\alpha_i^2 + \beta_i^2) = \sum_i c_i \alpha_i^{k'-t+2}\beta_i^t + \sum_i c_i \alpha_i^{k'-t}\beta_i^{t+2}.
	\]
	Then by induction, the we conclude that
	\[
	\sum_{i = 1}^{k + 1} c_i \alpha_i^{k'-t} \beta_i^t = \begin{cases}
	1 & \text{ if } t = k' - 1\\
	0 & \text{ otherwise}
	\end{cases}
	\]
    holds for all $0 \leq k' \leq k$ such that $k-k'$ is even. We are therefore almost done: we have established the \textbf{H1} case of Lemma~\ref{lem:summation} for all polynomials with only odd-degree terms or only even-degree terms.  Finally, for a general polynomial $f$ we can  decompose it as $f = f_{\text{odd}} + f_{\text{even}}$, where $f_{\text{odd}}$ (respectively, $f_{\text{even}}$) contains all the terms in $f$ with odd (respectively, even) degree. We know that there exist some vectors $\alpha, \beta, c$ and $\alpha', \beta', c'$ satisfying
	\[
	\odec{f}_{\text{odd}}(y, z) = \sum_{i=1}^{k+1} c_i f_{\text{odd}}(\alpha_iy + \beta_iz), \quad \odec{f}_{\text{even}} (y,z)= \sum_{i=1}^{k+1} c_i' f_{\text{even}}(\alpha_i'y + \beta_i'z),
	\]
	and $\|c\|_1, \|c'\|_1 \leq 20k$.
	Thus
	\begin{align*}
		\odec{f}(y, z) &= \odec{f}_{\text{odd}}(y, z) +\odec{f}_{\text{even}} (y,z)\\
		&= \sum_{i=1}^{k+1} c_i f_{\text{odd}}(\alpha_iy + \beta_iz) + \sum_{i=1}^{k+1} c_i' f_{\text{even}}(\alpha_i'y + \beta_i'z)\\
		&= \sum_{i =1}^{k+1}\frac12 c_i (f(\alpha_iy + \beta_iz) -f(-\alpha_iy-\beta_iz)) + \sum_{i=1}^{k+1} \frac12 c_i' (f(\alpha_i'y + \beta_i'z)+f(-\alpha_i'y-\beta_i'z))\\
		&= \sum_{i = 1}^{4(k+1)} c_i'' f(\alpha_i''y + \beta_i''z),
	\end{align*}
	where $c'' = (c/2, -c/2, c'/2, c'/2), \alpha'' = (\alpha, -\alpha, \alpha', -\alpha'), \beta''=(\beta, -\beta, \beta', -\beta')$ and $\|c''\|_1 \leq 40k$.
\end{proof}

\section*{Acknowledgments}
The authors would like to thank Oded Regev for helpful discussions, and John Wright for permission to include~\eqref{eqn:one-liner}.

\ignore{

	\section{Introduction}
	\subsection{Decoupling}
	
	In statitics, one general approach called \emph{decoupling} is widely used to for handling complicate problems with dependent variables. The decoupling principle establishes some inequalities between the problems with dependent variables and those introducing enough independency. After that those problems can be analyzed by the theory of independent random variables which is much easier.
	
	For any multilinear polynomial $f : \R^n \to \R$ of degree $k$,
	\[
	f(x) = a_0 + \sum_{t= 1}^k \sum_{i \in I_n^{t}} a_{i_1,\dots,i_{t}} x_{i_1}\dots x_{i_{t}}
	\]
	the fully decoupled version of $f$ is $\tilde{f} : \R^{n \times k} \to \R$
	\begin{align*}
		\tilde{f}(x^1, \dots, x^k) &= a_0 + \sum_{t = 1}^k \frac{(k - t)!}{t!}\frac{(n-k)!}{(n-t)!}\sum_{i \in I_n^k}\sum_{r\in I_k^{t}}a_{i_{r_1}, \dots, i_{r_t}} x_{i_{r_1}}^{r_1} \dots x_{i_{r_t}}^{r_t}\\
		&=a_0 + \sum_{t = 1}^k \frac{(k - t)!}{t!}\sum_{i \in I_n^k}\sum_{r\in I_k^{t}}a_{i_1, \dots, i_t} x_{1}^{r_1} \dots x_{t}^{r_t}
	\end{align*}

	Roughly speaking, $\tilde{f}$ is that for each monomial of $f$, $\tilde{f}$ will assign those $k$ variables with the correpsonding coordinate in $k$ indepedent copies  $x^1, \dots, x^k \in \R^n$ uniformly randomly, i.e., if $f(x) = x_1x_2x_3$, then the decoupled version
	\[
	\tilde{f}(x^1, x^2, x^3) = \frac16(x^1_1x^2_2x^3_3 + x^1_1x^3_2x^2_3+x^2_1x^1_2x^3_3+x^2_1x^3_2x^1_3+x^3_1x^1_2x^2_3+x^3_1x^2_2x^1_3)
	\]
	
	In \cite{kwapien1987}, they show that there is an equation between fully decoupled version and the original function
	\[
	\tilde{f}(x^1,\dots, x^k) = \E_{\epsilon_1, \dots, \eps_k \in \{-1,1\}}\left[\frac{\epsilon_1 \epsilon_2 \dots  \epsilon_k}{k!} (\epsilon_1 + \dots + \epsilon_k)^k f\left(\frac{\epsilon_1x^1+\dots+\eps_kx^k}{\epsilon_1+\dots+\epsilon_k}\right) \right]
	\]
	
	It is easy to check $\tilde{f}$ has following properties:
	
	\begin{enumerate}
		\item $\tilde{f}$ is symmetric, i.e.,$\tilde{f}(x^1,\dots, x^k)$ is invariant under the permutation of $x^1,\dots, x^k$.
		\item $\tilde{f}(x, \dots, x) = f(x)$ for any $x \in \R^n$.
	\end{enumerate}
	
	This formula is also a well-known ``polarization-formula''.
	
	In Section 3 of \cite{decoupling2012}, they list the following decoupling inequality for polynomials with low degree on convex functions and tail probability (see also \cite{kwapien1987} for homogenuous case with Gaussian/Boolean inputs).
	\begin{theorem}\label{thm:fulldecouple}
		For natural numbers $n \geq k$, let $\{X_i\}_{i=1}^n$ be $n$ independent random variables with values in a measurable space $S$, and let $\{X_i^j\}_{i=1}^n$, $j = 1, \dots, k$ be $k$ independent copies of this sequence. Let $B$ be a separable Banach space and, $f: \R^n \to B$ be a multilinear polynomial of degree $k$ with coefficients in Banach space $B$, and if $\tilde{f}$ is its fully-decoupled version,
		then there exists $C_k = \exp(O(k^2 \log k))$, such that
	\[
	 \E\left[\Phi\left(\frac1{C_k}\left\|f(X)\right\|\right)\right] \leq \E\left[\Phi\left(\left\|\tilde{f}(X^1, \dots, X^k)\right\|\right)\right] \leq \E\left[\Phi\left(C_k\left\|f(X)\right\|\right)\right]
	\]
	and for all $t > 0$,
	\[
	\frac1{C_k}\Pr\left[\frac1{C_k}\left\|f(X)\right\| >  t\right] \leq
	\Pr[\left\|\tilde{f}(X^1, \dots, X^k)\right\| >t] \leq
	C_k\Pr[C_k\left\|f(X)\right\| >  t].
	\]
	
	In particular, for homogenuous function and Gaussian/Boolean inputs, the constant $C_k$ can be reduced to $\exp(O(k))$.
	
	\end{theorem}
	
	Decoupling inequality for convex function for multilinear polynomial with degree 2 or in general has a long history, see \cite{burkholder1983geometric}, \cite{zinn1985comparison}, \cite{mcconnell1986decoupling}, \cite{mcconnell1987decoupling}, \cite{de1987decoupling}, \cite{kwapien1987}, \cite{bourgain1987invertibility} and \cite{de1994order}. \cite{kwapien1987} first showed a decoupling inequality for tail probability for homogeneous case with Gaussian inputs. A decoupling inequality for tail probability of multilinear polynomial for independent symmetric random variables  is found in \cite{kwapien1992random}. The decoupling inequality is generalized to U-process in \cite{de1992decoupling}, \cite{victor1993bounds} and \cite{de1995decoupling}.
	
	Recent work has shown the importance of block-multilinear Boolean function. Let $f:\{-1,1\}^n \to \R$ be a multilinear polynomial of degree $k$. We say $f$ is \emph{block-multilinear} if we can divide all $n$ variables into $k$ different blocks, so that every monomial of $f$ is the product of at most one variable from each block. In \cite{lovett2010elementary}, they prove an anti-concentration result for polynomials in Gaussian variables starting from block-multilinear case. In \cite{kane2013prg}, they use a random hash function to partition the variables to make a low degree polynomial block-multilinear, and they presents a pseudorandom generator that fools Lipschitz functions of polynomial based on it.In \cite{aaronson2014}, they present a classical randomized algorithm for block-multilinear functions that queries only $O(\frac1{\epsilon^2}n^{(1-1/k)})$ variables and can estimate the output of $f$ within error $\eps$.
	
	\begin{theorem}[Query complexity for block-mulitilinear functions]
		\label{thm:queryforblock}
		Let $f : \{-1,1\}^{n} \to [-1,1]$ be any block-multilinear polynomial of degree $k$ bounded in $[-1, 1]$. Then there exists a classical randomized algorithm that, on input $x \in \{-1,1\}^{n}$, non adaptively queries $O(\frac1{\eps^2}(2e)^{2k}n^{1-1/k})$ bits of $x$, and then outputs an estimate $f^*$ such that with high probability,
		\[
		|f^* - f(x)| \leq \epsilon
		\]
	\end{theorem}
	
	Notice that a fully-decoupled function $\tilde{f}$ is always block-multilinear. It is possible to use decoupling inequality to extend some results on block-multilinear polynomials to general multilinear polynomials with Boolean/Guassian inputs.
	
	However, as far as we know, there are few study about the application of decoupling principle in analysis of boolean function. In \cite{khot2008} and \cite{barak2015}, they both use the decoupling principle of degree-3 in their algorithm.
	
	On the other hand, the downside of fully decoupling is that it will involving an exponential factor on $k$. Notice that even if $f$ is a linear map on a subset of coordinates, there are still some good properties, i.e., it is easy to analyze the tail probability. For any multilinear polynomial $f : \R^n\to\R$,
	\[
	f(x) = \sum_{S \subseteq[n]} \widehat{f}(S) x^S,
	\]
	we define the \emph{one-variable decoupled} version of $f$ is $\breve{f} : \R^n \times \R^n \to \R$ and
	\[
	\breve{f}(y, z) = \sum_{\emptyset \neq S \subseteq [n]} \widehat{f}(S) \sum_{i \in S} y_i z^{S/\{i\}}
	\]
	where $y, z \in \R^n$.
	Here $\breve{f}(y, z)$ is a linear map on $y$. In this paper, we will show a decoupling inequality for one-variable decoupling, with a factor only linear of $k$. We use the decoupling inequality to improve the results in \cite{DFKO}.

	\subsection{Tail bound and Fourier Tails}
	
	In the analysis of Boolean function, a kind of functions we are interested in is those functions only rely on a small number of input, called ``juntas''.
	\begin{definition}
		A function $f: \{-1, 1\}^n \to \R$ is called an $(\epsilon, j)$-junta if there eixts a function $g:\{-1, 1\}^n \to \R$ depending on at most $j$ coordinates such that $\|f-g\|_2^2 \leq \epsilon$.
	\end{definition}
	In \cite{fri98}, \cite{friedgut2002boolean} and \cite{bourgain2002distribution}(see also \cite{khot2006nonembeddability}), they show some junta-related theorems related to the $L_2$ Fourier mass on charaters of high degree for Boolean functions with output $\{-1, 1\}$. A natural question is: Does a theorem of the same flavour holds real-valued, bounded funcitons, $f:\{-1, 1\}^n \to [-1, 1]$? The answer is Yes! In \cite{DFKO}, they show the following theorem:
	
	\begin{theorem}
		\label{thm:junta-DFKO}
		Let $f : \{-1, 1\}^n \to [-1, 1]$, $k\geq 1$, and $\epsilon > 0$. Suppose
		\[
		\sum_{|S| > k} \widehat{f}(S)^2 \leq \exp(-O(k^2\log k)/\eps)
		\]
		Then $f$ is an $(\epsilon, 2^{O(k)}/\epsilon^2)$-junta.
	\end{theorem}

	Whereas the main theorem in \cite{DFKO} concerned showing largeness of Fourier tails for non-juntas, this was essentially a corollary of their main technical theorem, for the tail probability of Boolean function with low degrees.
	
	\begin{theorem}
		\label{thm:tail-DFKO}
		There is a universal constant $C$ such that the following holds: Suppose $f:\{-1, 1\}^n \to \R$ has degree at most $k$ and that $\Var[f] = 1$. Let $t \geq 1$ and suppose that $\MaxInf[f] \leq t^{-2}C^{-k}$. Then
		\[
		\Pr[|f|\geq t] \geq \exp(-Ct^2k^2\log k)
		\]
	\end{theorem}
	
	This theorem is interesting in its own right. It has most often been applied in the study of quantum query complexity, e.g.\cite{aaronson2014}, \cite{barak2015}.
	
	In \cite{DFKO}, they also show these theorems are sharp up to the $\log k$ factor in the exponential.
	
	\begin{theorem}[Tightness of Theorem~\ref{thm:tail-DFKO} and \ref{thm:junta-DFKO}]
		For any $t \geq 1$, there exists $f:\{-1,1\}^n \to \R$ having degree at most $k$, satisfying $\sum_{S \neq \emptyset} \widehat{f}(S)^2 = 1$ and $\sum_{S \ni i} \widehat{f}(S)^2 \leq t^{-2}C^{-k}$ for all $i$, and yet still
		\[
		\Pr[|f| \geq t] \leq \exp(-Ct^2k^2).
		\]
		Furthermore, there exists some other function $g: \{-1, 1\}^n \to \R$ with
		\[
		\sum_{|S| > k} \widehat{g}(S)^2 \leq \exp(-\Omega(k^2))
		\]
		where $g$ is not even a $(c, n/2)$-junta for some universal small constant $c$.
	\end{theorem}
	
	These functions are constructed by involving Chebyshev polynomials. The key idea of the proof of Theorem~\ref{thm:tail-DFKO} is using random restriction to approximate the function to a linear map. Most of the Fourier weights is brought down to the first level under random restriction and then function is close to a linear map with  variance at least $1/O(\log k)$ fraction of the variance of original function with high probability. In this paper, we get rid of the $\log k$ gap between upper and lower bounds via one-variable decoupling.
		
	\subsection{Our result}
	We show a new decoupling inequality for one-variable decoupling with Gaussian inputs.
	\begin{theorem}
	\label{thm:one-variable}
		Let $\bX, \bY, \bZ \in N(0, 1)^n$. Let $f : \R^n \to \R$ be a multilinear polynomial with degree $k$ and $\breve{f}$ be the one-variable decoupled version of $f$.
		
		Then for any convex non-decreasing function $\Phi : [0, +\infty) \to [0, +\infty)$,
		\[
		\E_{\bY, \bZ}\left[\Phi\left(|\breve{f}(\bY, \bZ)|\right)\right] \leq \E_{\bX}[\Phi(Ck|f(\bX)|)],
		\]
		and
		\[
		\Pr_{\bY, \bZ}[|\breve{f}(\bY, \bZ)| \geq t] \leq 4(k+1)\Pr_{\bX}\left[Ck|f(\bX)| \geq t\right]
		\]
		for some universal constant $C > 0$.
	\end{theorem}
	Though Theorem~\ref{thm:one-variable} only holds for Gaussian inputs, it gives a better factor of $O(k)$ rather than $\exp(O(k))$ in fully decoupling. We will show the following application which is sharper than the bound given in \cite{DFKO}.
	
	\begin{theorem}
		\label{thm:tail}
		There is a universal constant $C$ such that the following holds: Suppose $f:\{-1, 1\}^n \to \R$ has degree at most $k$ and that $\Var[f] = 1$. Let $t \geq 1$ and suppose that $\MaxInf[f] \leq C^{-k^3t^2}$. Then
		\[
		\Pr[|f|\geq t] \geq \exp(-Ct^2k^2)
		\]
	\end{theorem}
	\begin{theorem}
		\label{thm:junta}
		Let $f : \{-1, 1\}^n \to [-1, 1]$, $k\geq 1$, and $\epsilon > 0$. Suppose
		\[
		\sum_{|S| > k} \widehat{f}(S)^2 \leq \exp(-O(k^2)/\eps)
		\]
		Then $f$ is an $(\epsilon, 2^{O(k^3)}/\epsilon^2)$-junta.
	\end{theorem}
	
	We get rid of $\log k$ gap between upper and lower bound in \cite{DFKO}. These theorems are sharp and can be achieved by some function given in \cite{DFKO}. The tightness of Theorem~\ref{thm:tail} and \ref{thm:junta} also indicate Theorem~\ref{thm:one-variable} is tight.
	
	We also have a minor application of fully decoupling, which prove a conjecture in \cite{aaronson2014} that a classical query algorithm for fully decoupled (block-multilinear) functions also works for other Boolean functions.
	
	\begin{theorem}
		\label{thm:query}
		Suppose $f : \{-1,1\}^{n} \to [-1,1]$ has degree at most $k$. Then there exists a classical randomized algorithm that, on input $x \in \{-1,1\}^{n}$, non adaptively queries $O(\frac1{\eps^2}C^kn^{1-1/k})$ bits of $x$, where $C$ is a universal constant and then outputs an estimate $f^*$ such that with high probability,
		\[
		|f^* - f(x)| \leq \epsilon.
		\]
	\end{theorem}
	
	\subsection{Orgainzation}
		Section~\ref{sec:prelimiaries} contains definitions and notations. Section~\ref{sec:one-variable} shows the one-variable decoupling inequalities, proving Theorem~\ref{thm:one-variable}. We prove Theorem~\ref{thm:tail} and \ref{thm:junta} in Section~\ref{sec:tailbound}, as an application of one-variable decoupling. We show another application for fully decoupling in Section~\ref{sec:application2}, proving Theorem~\ref{thm:query}. A bug fixing of a lemma in \cite{aaronson2014} can be found in Appendix~\ref{sec:appendix}.	
	\section{Preliminaries}
	\label{sec:prelimiaries}
	\subsection{Basic definitions and decoupling}
	We deonte $[n]$ as $\{1, \dots, n\}$. We denote
	\[
	I_n^k = \{(i_1, \dots, i_k : i_j \in \N, \leq i_j \leq n, i_{j_1} \neq i_{j_2} \text{ if } j_1 \neq j_2)\}
	\]
	We define $x^{S} = \prod_{i \in S} x_i$ for any $S \subseteq [n]$, and we express a multilinear polynomial $f : \R^n \to \R$ as
	\[
	f(x) = \sum_{S \subseteq [n]} \widehat{f}(S) x^S.
	\]
	Where $\widehat{f}(S)$ are real numbers. If we only focus on Boolean inputs, saying $x \in \{-1,1\}^n$, this expression is also known as the Fourier transform of $f$. It is easy to check that polynomial $f$ is an odd function if and only if $\widehat{f}(S) = 0$ for all $|S|$ even and 	$f$ is an even function if and only if $\widehat{f}(S) = 0$ for all $|S|$ odd. We introduce the notation
	\[
	\Var[f] = \sum_{S \neq \emptyset} \widehat{f}(S)^2 \qquad \Inf_i[f] = \sum_{S \ni i} \widehat{f}(S)^2 \qquad \MaxInf[f] = \max_i \Inf_i[f].
	\]
	
	Note that $\Var[f]$ is the variance of function $f$ on Gaussian space $N(0,1)^n$ or Boolean cube $\{-1, 1\}^n$.
	
	We say $f$ is homogeneous with degree $k$ if $\widehat{f}(S) = 0$ for all $|S| \neq k$.
	
	\begin{definition}[One-variable decoupling]\label{def:one-variable}
		For any multilinear polynomial $f : \R^n\to\R$,
		\[
		f(x) = \sum_{S \subseteq[n]} \widehat{f}(S) x^S,
		\]
		the \emph{one-variable decoupled} version of $f$ is $\breve{f} : \R^n \times \R^n \to \R$ and
		\[
		\breve{f}(y, z) = \sum_{\emptyset \neq S \subseteq [n]} \widehat{f}(S) \sum_{i \in S} y_i z^{S/\{i\}}
		\]
		where $y, z \in \R^n$.
	\end{definition}
	
	Here each term of $\breve{f}$ contains only one variable from $y$, so the one-variable decoupled version $\breve{f}(y, z)$ is a linear function on vector $y$. However unlike the full decoupling, $\breve{f}(x, x) = \sum_{S \subseteq[n]}|S|\widehat{f}(S)\prod_{i\in S} x_i \neq f(x)$.
	
	We can also express the mulitilinear polynomial $f$ of degree $k$ as
	\[
	f(x) = a_0 + \sum_{t = 1}^k \sum_{i \in I_n^{t}} a_{i_1,\dots,i_{t}} x_{i_1}\dots x_{i_{t}}
	\]
	where $a_i$ are symmetric in their subindices, that is $a_{i_1, \dots, i_{t}} = a_{\sigma({i_1}), \dots, \sigma({i_{t})}}$ for any permutation $\sigma$ on $1, \dots, t$. Or that is to say,
	\[
	a_0 = \widehat{f}(\emptyset), \quad a_{i_1, \dots, i_{t}} = \frac1{t!} \widehat{f}(\{i_1, \dots, i_{t}\})
	\]
	Then we define fully decoupling as following:
	\begin{definition}[Fully decoupling]
		For any multilinear polynomial $f : \R^n \to \R$ of degree $k$,
		\[
		f(x) = a_0 + \sum_{t= 1}^k \sum_{i \in I_n^{t}} a_{i_1,\dots,i_{t}} x_{i_1}\dots x_{i_{t}}
		\]
		the fully decoupled version of $f$ is $\tilde{f} : \R^{n \times k} \to \R$
		\begin{align*}
		\tilde{f}(x^1, \dots, x^k) &= a_0 + \sum_{t = 1}^k \frac{(k - t)!}{t!}\frac{(n-k)!}{(n-t)!}\sum_{i \in I_n^k}\sum_{r\in I_k^{t}}a_{i_{r_1}, \dots, i_{r_t}} x_{i_{r_1}}^{r_1} \dots x_{i_{r_t}}^{r_t}\\
		&=a_0 + \sum_{t = 1}^k \frac{(k - t)!}{t!}\sum_{i \in I_n^k}\sum_{r\in I_k^{t}}a_{i_1, \dots, i_t} x_{1}^{r_1} \dots x_{t}^{r_t}
		\end{align*}
	\end{definition}
	
	\subsection{Properities of Boolean/Gaussian function}
	In this paper, we will use $\bx, \by, \bz \dots$ to denote uniformly random variable on the Boolean cube $\{-1, 1\}^n$, and $\bX, \bY, \bZ$ to denote random standard Gaussian vectors in distribution $N(0,1)^n$.Suppose $\bx^1, \dots, \bx^M \sim \{-1, 1\}^n$ are independent random Boolean vectors. $\bX = \frac{1}{\sqrt{M}} (\bx^1+ \dots + \bx^M)$ is close to the standard Gaussian distribution $N(0,1)^n$ when $M \to \infty$. Since we can use the sum of a large number of random Boolean bits to ``simulate'' Gaussian inputs, multilinear polynomial functions on Gaussian space is a special case of Boolean function. A further key point is this ``simulating'' technique preserves polynomial degree. That is, if $f(\bX)$ is a multilinear polynomial degree-$k$ and we substitute $\bX = \frac{1}{\sqrt{M}} (\bx^1+ \dots + \bx^M)$, we get another multilinear polynomial $f'(\bx_1, \dots, \bx_M) = f\frac{1}{\sqrt{M}} (\bx_1+ \dots + \bx_M)$ where the degree of $f'$ is still $k$. Therefore the following lemma (Theorem~9.24 from \cite{Ryan2014}) holds for both Boolean and Gaussian input.
	\begin{lemma}
		\label{lem:exceedmean}
		Let $f : \R^n \to \R$ be a polynomial of degree at most $k$. Then
		\[
		\Pr_{\bx \sim \{-1,1\}^n}\left[f(\bx) \geq \E[f]\right] \geq \frac14 e^{-2k}
		\qquad
		\Pr_{\bX \sim N(0,1)^n}\left[f(\bX) \geq \E[f]\right] \geq \frac14 e^{-2k}
		\]
	\end{lemma}
	
	On the other hand, \cite{mossel2005noise} showed that any low-degree Boolean function with small influences can be approximable if we switch the input to a Gaussian space.
	
	\begin{lemma}[Invariance Principle]
	\label{lem:invariance}
	Assume $f : \R^n \to \R$ is multilinear polynomial satisfying $\Var[f] = 1$, $\deg(f) \leq k$ and $\MaxInf(f) \leq \tau$. Then
	\[
	\left|\Pr_{\bx \sim \{-1,1\}^n}[f(\bx) \geq t] - \Pr_{\bX \sim N(0,1)^n}[f(\bX) \geq t]\right| \leq O(k\tau^{1/(4k+1)})
	\]
	\end{lemma}
	
	\section{One-variable decoupling}\label{sec:one-variable}
	We will prove three variants of one-variable decoupling under some similar hypotheses. Specifically, they will be concerned with a multilinear polynomial $f$ over random variables $\bX$ and its one-varible-decoupled version $\breve{f}$ over random variables $\bY, \bZ$. We will lay out three hypotheses as following:
	\begin{itemize}
		\item[\textbf{H1}] Let $\bX, \bY, \bZ$ be independent length-$n$ vectors where each $\bX_i$(respectively $\bY_i$, $\bZ_i$) is a standard Gaussian $N(0,1)$ independently.
		\item[\textbf{H2}] Let $\bX, \bY, \bZ$ be independent length-$n$ vectors where each $\bX_i$(respectively $\bY_i$, $\bZ_i$) is a uniformly random $\pm 1$ bit independently. Let $f$ be homogeneous.
		\item[\textbf{H3}] Let $\bX, \bY, \bZ$ be independent length-$n$ vectors where each $\bX_i$(respectively $\bY_i$, $\bZ_i$) is a uniformly random $\pm 1$ bit independently.
	\end{itemize}
	
	\begin{theorem}
		\label{thm:one-variable-full}
		Let $\bX, \bY, \bZ \sim N(0, 1)^n$. Let $f : \R^n \to \R$ be a multilinear polynomial with degree $k$ and $\breve{f}$ be the one-variable decoupled version of $f$.
		
		Then for any convex non-decreasing function $\Phi : [0, +\infty) \to [0, +\infty)$,
		\[
		\E_{\bY, \bZ}\left[\Phi\left(|\breve{f}(\bY, \bZ)|\right)\right] \leq \E_{\bX}[\Phi(C_k|f(\bX)|)],
		\]
		and for tail bound of any $t \geq 0$,
		\[
		\Pr_{\bY, \bZ}[|\breve{f}(\bY, \bZ)| \geq t] \leq D_k \Pr_{\bX}\left[C_k |f(\bX)| \geq t\right]
		\]
		where
		\[
		C_k = \begin{cases}
		O(k) & \text{under \textbf{H1}},\\
		O(k^{3/2}) &  \text{under \textbf{H2}},\\
		O(k^2) &  \text{under \textbf{H3}},
		\end{cases}
		\]
		and
		\[
		D_k = \begin{cases}
		O(k) & \text{under \textbf{H1}},\\
		exp(O(k \log k)) &  \text{under \textbf{H2, H3}}.\\
		\end{cases}
		\]
	\end{theorem}
	
	The key idea of the proof is to express $\breve{f}(y, z)$ as a summation of $f$ with inputs a combination of $y$ and $z$ with different weights:
	\[
	\breve{f}(y, z) = \sum_{i=1}^{m} c_i f(a_i y + b_i z).
	\]
	We also need some constraints on $a_i$ and $b_i$ to maintain Gaussian/Boolean properties.
	\begin{lemma}
		\label{lem:summation}
		For any integer $k$, there exist some vectors $a, b, c \in \R^{m}$ such that for any multilinear polynomial $f$ with degree $k$ and the one-variable decoupled version $\breve{f}$ corresponding to $f$,
		\[
		\breve{f}(y, z) = \sum_{i=1}^{m} c_if(a_iy + b_iz).
		\]
		satisfying $\|c\|_1 \leq C_k$, $m = O(k)$ and with constraints
		\begin{align*}
			a_i^2 + b_i^2 = 1 & \quad \text{under \textbf{H1}},\\
			|a_i| + |b_i| = 1 \text{ and } |a_i|, |b_i| \geq 1/O(k^{3/2}) & \quad \text{under \textbf{H2}},\\
			|a_i| + |b_i| = 1 \text{ and } |a_i|, |b_i| \geq 1/O(k^{2}) & \quad \text{under \textbf{H3}}.\\
		\end{align*}
	\end{lemma}
	
	With Lemma~\ref{lem:summation}, the proof of Theorem \ref{thm:one-variable-full} is quite straightforward. The Boolean case is slight more twisted, so we will prove Gaussian case first.
	
	\begin{proof}[Proof of Theorem~\ref{thm:one-variable-full}, Hypothesis \textbf{H1}]
		By Lemma~\ref{lem:summation}, for any convex non-decreasing function $\Phi : [0, +\infty) \to [0, +\infty)$,
		\begin{align*}
			\E_{\bY, \bZ }\left[\Phi\left(|\breve{f}(\bY, \bZ)|\right)\right] &= \E_{\bY, \bZ}\left[\Phi\left(\left|\sum_{i = 1}^{m}c_i f(a_i\bY+b_i\bZ)\right|\right)\right]\\
			&\leq \E_{\bY, \bZ}\left[\Phi\left(\sum_{i = 1}^{m}|c_i|| f(a_i\bY+b_i\bZ)|\right)\right]\\
			&\leq \sum_{i = 1}^{m}\frac{|c_i|}{\|c\|_1} \E_{\bY, \bZ}\left[\Phi\left(\|c\|_1 |f(a_i \bY + b_i \bZ)|\right)\right]\\
			&= \sum_{i=1}^{m}\frac{|c_i|}{\|c\|_1} \E_{\bX}[\Phi(\|c\|_1 |f(\bX)|)]\\
			&\leq \E_{\bX}[\Phi(C_k|f(\bX)|]
		\end{align*}
		The inequalities hold from the convexity and monotonicity of $\Phi$.	The second equality holds since $a_i \bY + b_i \bZ \sim N(0,1)^n$ under the constraint $a_i^2 + b_i^2 = 1$.
		
		For the tail bound, by Lemma~\ref{lem:summation}, for any $y, z$ satisfying $|\breve{f}(y, z)| \geq t$, there exists some $i$ such that $\|c\|_1 |f(a_i y + b_i z)| \geq t$. Then by Pigeonhole principle, there exists some $i^* \in [m]$ satisfying
		\[
		\Pr_{\bY, \bZ}\left[\|c\|_1 |f(a_{i^*} \bY +b_{i^*} \bZ)| \geq t \, | \, |\breve{f}(\bY, \bZ)| \geq t \right] \geq \frac1{m}.
		\]
		
		Therefore
		\begin{align*}
			\Pr_{\bX} [C_k|f(\bX)| \geq t] &\geq \Pr_{\bY, \bZ} [\|c\|_1 |f(a_{i^*}\bY + b_{i^*}\bZ)| \geq t]\\
			&\geq \Pr_{\bY, \bZ}\left[\|c\|_1 |f(a_{i^*} \bY +b_{i^*} \bZ)|\right] \geq t \, | \, |\breve{f}(\bY, \bZ)| \geq t] \Pr_
			{\bY, \bZ} [|\breve{f}(\bY, \bZ)| \geq t]\\
			&\geq \frac1{m} \Pr_{\bY, \bZ} [|\breve{f}(\bY, \bZ)| \geq t]
		\end{align*}
		
	\end{proof}
	
	\begin{proof}[Proof of Theorem~\ref{thm:one-variable-full}, Hypothesis \textbf{H2, H3}]
		We define random bits
		\[
		\bX^{(i)}_j = \begin{cases}
		\bY_j &\quad \text{with probability } |a_i|,\\
		\bZ_j &\quad \text{with probability } |b_i|,
		\end{cases}
		\]
		for all $i \in [m]$ and $j \in [n]$ independently. Notice since $\bY, \bZ$ are uniformly random from $\{-1, 1\}^n$, $\bX^{(i)}$ is also uniformly random in $\{-1, 1\}^n$. By Lemma~\ref{lem:summation}, for any convex non-decreasing function $\Phi : [0, +\infty) \to [0, +\infty)$,
		\begin{align*}
			\E_{\bY, \bZ }\left[\Phi\left(|\breve{f}(\bY, \bZ)|\right)\right] &= \E_{\bY, \bZ}\left[\Phi\left(\left|\sum_{i = 1}^{m}c_i f(a_i\bY+b_i\bZ)\right|\right)\right]\\
			&\leq \E_{\bY, \bZ}\left[\Phi\left(\sum_{i = 1}^{m}|c_i|| f(a_i\bY+b_i\bZ)|\right)\right]\\
			&\leq \sum_{i = 1}^{m}\frac{|c_i|}{\|c\|_1} \E_{\bY, \bZ}\left[\Phi\left(\|c\|_1 |f(a_i \bY + b_i \bZ)|\right)\right]\\
			&\leq \sum_{i=1}^{m}\frac{|c_i|}{\|c\|_1} \E_{\bX}[\Phi(\|c\|_1 |f(\bX)|)]\\
			&\leq \E_{\bX}[\Phi(C_k|f(\bX|)]
		\end{align*}
		The first, second and last inequality holds as in hypothesis \textbf{H1}. The third inequality holds because $f(a_i \bY + b_i \bZ) = \E_{\bX^{(i)} \sim \bY, \bZ}[f(\bX^{(i)})]$ when $f$ is multilinear and then
		
		\begin{align*}
			\E_{\bY, \bZ}\left[\Phi\left(\|c\|_1 |f(a_i \bY + b_i \bZ)|\right)\right]
			&=\E_{\bY, \bZ}\left[\Phi\left(\|c\|_1 \left|\E_{\bX^{(i)}} [f(\bX^{(i)})]\right|\right)\right]\\
			&\leq \E_{\bY, \bZ, \bX^{(i)}}\left[\Phi\left(\|c\|_1 \left| f(\bX^{(i)})\right|\right)\right]\\
			&=\E_{\bX}[\Phi(\|c\|_1 |f(\bX)|)].
		\end{align*}

		For tail bound inequality, similar as the proof with hypothesis \textbf{H1},  there exists some $i^* \in [m]$ satisfying
		\[
		\Pr_{\bY, \bZ}\left[\|c\|_1 |f(a_{i^*} \bY +b_{i^*} \bZ)| \geq t\right] \geq \frac1{m}\Pr_{\bY, \bZ} [|\breve{f}(\bY, \bZ)| \geq t].
		\]
		Notice if $a_{i^*} = 1$ and $b_{i^*} = 0$, then $\bX^{(i^*)} = \bY$ then the result is trivial. Otherwise, we have $|a_{i^*}|, |b_{i^*}| \leq 1/O(k^2)$. For any fixed $\bY$ and $\bZ$, $\bX_j^{(i^*)}$ is a constant value if $\bY_j = \bZ_j$, and $\bX_j^{(i^*)}$ is a $\lambda$-biased distribution on $\{-1, 1\}$ with $\lambda \geq 1/O(k^2)$ if $\bY_j \neq \bZ_j$. Then by Lemma~\ref{lem:exceedmean-general}, for any fixed $\bY, \bZ$, we have
		\[
		\Pr_{\bX^{(i^*)}}\left[|f(\bX^{(i^*)})| \geq
		|f(a_{i^*}\bY + b_{i^*}\bZ)| \right]
		=
		\Pr_{\bX^{(i^*)}}\left[ |f(\bX^{(i^*)})| \geq  \left| \E_{\bX^{(i^*)}}\left[f(\bX^{(i^*)})  \right]\right| \right] \geq \frac1{O(\exp(k\log k))}		
		\]	
		Therefore,
		\begin{align*}
			\Pr_{\bX}[C_k |f(\bX)| \geq t] &\geq \Pr_{\bY, \bZ, \bX^{(i^*)}}\left[|f(\bX^{(i^*)})| \geq
			|f(a_{i^*}\bY + b_{i^*}\bZ)| \, \bigg| \, \|c\|_1 |f(a_{i^*}\bY + b_{i^*}\bZ)| \geq t \right] \\
			& \qquad \qquad \qquad \Pr_{\bY, \bZ} [\|c\|_1 |f(a_{i^*} \bY +b_{i^*} \bZ)| \geq t]\\
			& \geq \frac1{O(\exp(k\log k))} \Pr_{\bY, \bZ} [|\breve{f}(\bY, \bZ)| \geq t]
		\end{align*}
	\end{proof}
	
	The proof of Lemma~\ref{lem:summation} is choosing $a, b$ smartly to minimize $\|c\|_1$. We will show that it can be done by setting the ratio of $a_i$ and $b_i$ to be a (hyper)harmonic progression.
	
	\begin{proof}[Proof of Lemma~\ref{lem:summation}]
		We first consider the homogeneous case of degree $k$. We will show that for homogeneous case, $\breve{f}$ can be written as a sum of $k + 1$ different terms of $f$.
		\[
		\breve{f}(y, z) = \sum_{i=1}^{k+1} c_if(a_iy + b_iz).
		\]
		If we compare two sides term by term, it is equivalent to say for any $S \subseteq [n]$ with size $k$,
		\[
		\sum_{j \in S} y_j z^{S/\{j\}} = \sum_{i = 1}^{k + 1} c_i \prod_{j \in S} (a_i y_j + b_i z_j).
		\]
		Furthermore, we can simplify the constraints as
		\[
		\sum_{i = 1}^{k + 1} c_i a_i^{k-t} b_i^t = \begin{cases}
		1 & \text{ if } t = k - 1\\
		0 & \text{ otherwise}
		\end{cases}
		\]
		for all integers $ 0 \leq t \leq k$.
		
		If we define the ratio $\Delta_i = \frac{b_i}{a_i}$ and Vandermonde matrix V as
		\[
		V =
		\begin{bmatrix}
		1 & 1 & \dots & 1\\
		\Delta_1 & \Delta_2 &  \dots& \Delta_{k+1}\\
		\dots & \dots & \dots & \dots\\
		\Delta_1^{k-1} & \Delta_2^{k-1} &  \dots & \Delta_{k+1}^{k-1}\\
		\Delta_1^k & \Delta_2^k &  \dots& \Delta_{k+1}^k
		\end{bmatrix},
		\]
		Diagonal matrix $A = \text{diag}(a_1^k, a_2^k, \dots, a_{k+1}^k)$ and the indicator for $k$-th coordinate $e_k = (0, 0, \dots, 0, 1, 0)^{T}$. Then the matrix form of the constraints is:
		\[
		VAc = e_{k}.
		\]
		If we assume all $\Delta_i$'s are different, then $V$ is invertible. There is a formula for computing the inverse of Vandermonde matrix, so we can get the formula of $c_i$ as
		\[
		c_i = (A^{-1}V^{-1}e_{k})_i =  \frac{1}{a_i^k} \cdot \frac{\Delta_i - \sum_{j=1}^{k+1} \Delta_j}{\prod_{j=1, j \neq i}^{k+1} (\Delta_i - \Delta_j) }
		\]
		for all $1 \leq i \leq {k+1}$.
		
		If $k$ is odd, for the sake of convenience, we will replace our index $i$ from $1$ through $k+1$ by $1, 2, \dots, (k-1)/2$, $-1, -2, \dots, -(k-1)/2$ and $1/2, -1/2$ (Here $1/2$ and $-1/2$ are weird but we need them for easier analysis of the bound of  $|c_i|$). For Hypothesis \textbf{H1}, we assign
		\[
		\Delta_i = \frac{k}{i}, \quad a_i = \frac{i}{\sqrt{k^2 + i^2}}, \quad b_i = \frac{k}{\sqrt{k^2 + i^2}}
		\]
		for all $i\in \{1, 2, \dots, (k-1)/2, -1, -2, \dots, -(k-1)/2, 1/2, -1/2\}$. Then we know that $a_i^2 + b_i^2 = 1$, and the rest is to prove that $\|c\|_1 \leq O(k)$.
		
		For $1 \leq i \leq \frac{k-1}2$, we have
		\begin{align*}
			\left|\prod_{j \neq i} (\Delta_i - \Delta_j) \right|
			&= (\Delta_{1/2} - \Delta_i) (\Delta_i-\Delta_{-1/2})\left| \prod_{j=1, j \neq i}^{(k-1)/2} (\Delta_i - \Delta_j) \right| \prod_{j=-(k-1)/2}^{-1} (\Delta_i - \Delta_j) \\
			&= \left(2k - \frac{k}i\right)\left(2k + \frac{k}i\right) \prod_{j=1, j \neq i}^{(k-1)/2} \left|\frac{k}{i} - \frac{k}{j} \right| \prod_{j=1}^{(k-1)/2} (\frac{k}i + \frac{k}j) \\
			&= k^{k}\left(4-\frac1{i^2}\right) \prod_{j=1, j \neq i}^{(k-1)/2} \frac{|j-i|}{ij} \prod_{j=1}^{(k-1)/2}\frac{j+i}{ij}\\
			&= \frac{k^{k}}{i^{k-2}} \left(4-\frac1{i^2}\right) \frac{\left(\frac{k-1}2 + i\right)! \left(\frac{k-1}2 - i\right)!}{ \left[ \left(\frac{k-1}2 \right)!\right]^2}
		\end{align*}
		
		Then by the formula of $c_i$
		\begin{align*}
			|c_i| &=  \cdot \frac{1}{|a_i^k|} \cdot |\Delta_i - \sum_j \Delta_j| \cdot \frac1{\prod_{j \neq i} |\Delta_i - \Delta_j|} \\
			&= \cdot \left(\frac{\sqrt{k^2+i^2}}{i}\right)^k \cdot \frac{k}{i} \cdot \frac{i^{k-2}}{k^k} \cdot  \frac1{4-1/i^2} \cdot \frac{ \left[ \left(\frac{k-1}2 \right)!\right]^2}{\left(\frac{k-1}2 + i\right)! \left(\frac{k-1}2 - i\right)!}\\
			&= \frac{k}{i^3} \left(1 + \frac{i^2}{k^{2}}\right)^{k/2} \frac1{4-1/i^2} \cdot \frac{ \left[ \left(\frac{k-1}2 \right)!\right]^2}{\left(\frac{k-1}2 + i\right)! \left(\frac{k-1}2 - i\right)!}
		\end{align*}
		When $ 1 \leq i \leq \sqrt{k}$, we have
		\begin{align*}
			|c_i| &= \frac{k}{i^3} \left(1 + \frac{i^2}{k^{2}}\right)^{k/2} \frac1{4-1/i^2} \cdot \frac{ \left[ \left(\frac{k-1}2 \right)!\right]^2}{\left(\frac{k-1}2 + i\right)! \left(\frac{k-1}2 - i\right)!}\\
			&\leq \frac{k}{i^3} \left(1 + \frac{1}{k}\right)^{k/2}\\
			&\leq \frac{\sqrt{e}k}{i^3}.
		\end{align*}
		For $\sqrt{k} \leq i \leq \frac{k-1}2$, consider the ratio between $(i+1)^3|c_{i+1}|$ and $i^3 |c_i|$,
		\begin{align*}
			\frac{(i+1)^3|c_{i+1}|}{i^3|c_i|} &\leq \frac{(k^{2} + (i + 1)^2)^{k/2}}{(k^{2}+i^2)^{k/2}} \cdot \frac{\frac{k-1}2 - i}{\frac{k-1}2 + i +1} \\
			&=\left(1+\frac{2i+1}{k^2+i^2}\right)^{k/2}\cdot \frac{k-1-2i}{k+1+2i}\\
			&\leq \left(1+\frac{2i+1}{k^2}\right)^{k/2} \cdot \frac{k-1-2i}k\\
			&\leq e^{\frac{2i+1}{2k}}\left(1 - \frac{2i+1}k\right) \leq 1.
		\end{align*}
		The last inequality holds since $e^{x/2}(1-x) \leq 1$ for all $0 \leq x \leq 1$. This means $(i+1)^3|c_{i+1}| \leq i^3|c_i|$. Then by induction, we know that
		\[
		|c_i| \leq \frac{\sqrt{e}k}{i^3}
		\]
		for all $1 \leq i \leq (k-1)/2$.
		
		The last coefficient we need to analyze is $c_{1/2}$, and similarly we have
		\begin{align*}
			\left|\prod_{j \neq *} (\Delta_{1/2} - \Delta_j) \right|
			&= (\Delta_{\frac12} - \Delta_{-1/2}) \prod_{j=1}^{(k-1)/2} (\Delta_{1/2} - \Delta_j) \prod_{j=-(k-1)/2}^{-1} (\Delta_{\frac12} - \Delta_j) \\
			&= 4k  \prod_{j=1, }^{(k-1)/2}  \left(2k - \frac{k}j \right) \prod_{j=1}^{(k-1)/2} (2k + \frac{k}j) \\
			&= 4k^{k} \prod_{j=1}^{(k-1)/2} \frac{2j-1}{j} \prod_{j=1}^{(k-1)/2}\frac{2j+1}{j}\\
			&= 4k^{k}  \frac{(k-2)!!k!!}{ \left[ \left(\frac{k-1}2 \right)!\right]^2}
		\end{align*}
		Then we get
		\begin{align*}
			|c_{1/2}| &= \frac{1}{a_{1/2}^k} \cdot (\Delta_{1/2} - \sum_j \Delta_j) \cdot \frac1{\prod_{j\neq *} |\Delta_{1/2} - \Delta_j|} \\
			&=  \frac{(\sqrt{k^2+1/2^2})^k}{(1/2)^k} \cdot 2k \cdot \frac{1}{4k^{k}} \cdot  \frac{ \left[ \left(\frac{k-1}2 \right)!\right]^2}{(k-2)!!k!!}\\
			& = \left(1 + \frac{1}{4k}\right)^{k/2} \left(\frac{(k-1)!!}{(k-2)!!}\right)^2\\
			& \leq  e^{1/8} \cdot 2(k-1) \leq 2e^{1/8}k.
		\end{align*}
		
		It is easy to see that these coefficients are symmetric. $|c_i| = |c_{-i}|$ for all $1 \leq i \leq \frac{k-1}2$ or $i = \frac12$, so we conclude that
		\[
		\sum_i |c_i| = 2 \sum_{i = 1}^{(k-1)/2}|c_i| + 2|c_{1/2}|\leq 2\sqrt{e}\sum_{i =1}^{(k-1)/2} \frac{k}{i^3} + 4e^{1/8}k \leq 20k
		\]
		
		If $k$ is even, we define $\Delta_i, a_i, b_i$ for $i = 1, \dots, (k-2)/2, -1, \dots, -(k-2)/2$ and $1/2, -1/2$ similar to the odd case $k - 1$, with one extra coordinate $a_0 = 1$, $b_0 = \Delta_0 = 0$. It is easy to check that $c_0 = 0$ and the proof of bounds for other $|c_i|$ still holds.
		
		Following a similar analysis as above, we can also show that if we set
		\[
		\Delta_i = \frac{k^{3/2}}{i}, \quad a_i = \frac{i}{k^{3/2}+|i|}, \quad b_i = \frac{k^{3/2}}{k^{3/2} + |i|},
		\]
		for all $i\in \{1, 2, \dots, (k-1)/2, -1, -2, \dots, -(k-1)/2, 1/2, -1/2\}$ (Assume $k$ is odd, and use the same trick as above for even $k$), then we have $|a_i| + |b_i| = 1$ and $\|c\|_1 \leq O(k^{3/2})$ (for Hypothesis \textbf{H2}). And if we set
		\[
		\Delta_i = \frac{k^{2}}{i^2}, \quad a_i = \frac{i^2}{k^{2}+i^2}, \quad b_i = \frac{k^{2}}{k^{2} + i^2},
		\]
		for all $i\in \{1, 2, \dots, k, \text{ and } 1/2\}$, then we have $a_i + b_i = 1$ and $\|c\|_1 \leq O(k^{2})$ (for Hypothesis \textbf{H3}).
		
		Now we are done for Hypothesis \textbf{H2}, since we assume $f$ is homogeneous under \textbf{H2}. For Hypothesis \textbf{H3},  a key observation is that, since $a_i + b_i = 1$ for all $i$,
		\[
		\sum_{i} c_i a_i^{k' - t} b_i^{t} = \sum_{i} c_i a_i^{k' - t} b_i^t (a_i + b_i) = \sum_i c_i a_i^{k'-t+1}b_i^t + \sum_i c_i a_i^{k'-t}b_i^{t+1}.
		\]
		for any $k' < k$.
		Then by induction, the same configuration of homogeneous case satisfies for the constraints
		\[
		\sum_{i} c_i a_i^{k'-t} b_i^t = \begin{cases}
		k-k' & \text{ if } t = k'\\
		1 & \text{ if } t = k' - 1\\
		0 & \text{ otherwise}
		\end{cases}
		\]
		for all $k' \leq k$. Now the case $t = k'$ is different from what we want, but an easy trick can get rid of it:
		\[
		\frac12\sum_{i} c_i a_i^{k'-t} b_i^t - \frac12\sum_{i} c_i (-a_i)^{k'-t} b_i^t = \frac{1 - (-1)^{k'-t}}2 \sum_{i} c_i a_i^{k'-t} b_i^t = \begin{cases}
		1 & \text{ if } t = k' - 1\\
		0 & \text{ otherwise}
		\end{cases}
		\]
		Therefore,
		\[
		\breve{f}(y, z) = \sum_{i = 1}^{m}c_i f(a_i y + b_i z)
		\]
		for general multilinear polynomial $f$ with degree $k$ with $|c| \leq O(k^2)$, $|a_i| + |b_i| = 1$ and $m = 2(k+1)$.
		
		For hypothesis \textbf{H1}, similarly, $a_i^2 + b_i^2 = 1$ for all $i$,
		\[
		\sum_{i} c_i a_i^{k' - t} b_i^{t} = \sum_{i} c_i a_i^{k' - t} b_i^t (a_i^2 + b_i^2) = \sum_i c_i a_i^{k'-t+2}b_i^t + \sum_i c_i a_i^{k'-t}b_i^{t+2}.
		\]
		Then by induction, the same configuration of homogeneous case satisfies for the constraints
		\[
		\sum_{i = 1}^{k + 1} c_i a_i^{k'-t} b_i^t = \begin{cases}
		1 & \text{ if } t = k' - 1\\
		0 & \text{ otherwise}
		\end{cases}
		\]
		for any $0 \leq k' \leq k$ satisfying $k-k'$ is even. Therefore Lemma~\ref{lem:summation} holds for all polynomials with terms in only odd or even degrees. For a general degree-$k$ polynomial $f$, we can  decompose $f = f_{\text{odd}} + f_{\text{even}}$, where $f_{\text{odd}}, f_{\text{even}}$ are terms of $f$ with only odd or even degrees. Then we know that there exist some vectors $a, b, c$ and $a', b', c'$ satisfying
		\[
		\breve{f}_{\text{odd}}(y, z) = \sum_{i=1}^{k+1} c_i f_{\text{odd}}(a_iy + b_iz), \quad \breve{f}_{\text{even}} (y,z)= \sum_{i=1}^{k+1} c_i' f_{\text{even}}(a_i'y + b_i'z)
		\]
		while $\|c\|_1, \|c'\|_1 \leq 20k$.
		Therefore
		\begin{align*}
			\breve{f}(y, z) &= \breve{f}_{\text{odd}}(y, z) +\breve{f}_{\text{even}} (y,z)\\
			&= \sum_{i=1}^{k+1} c_i f_{\text{odd}}(a_iy + b_iz) + \sum_{i=1}^{k+1} c_i' f_{\text{even}}(a_i'y + b_i'z)\\
			&= \sum_{i =1}^{k+1}\frac12 c_i (f(a_iy + b_iz) -f(-a_iy-b_iz)) + \sum_{i=1}^{k+1} \frac12 c_i' (f(a_i'y + b_i'z)+f(-a_i'y-b_i'z))\\
			&= \sum_{i = 1}^{4(k+1)} c_i'' f(a_i''y + b_i''z),
		\end{align*}
		
		where $c'' = (c/2, -c/2, c'/2, c'/2), a'' = (a, -a, a', -a'), b''=(b, -b, b', -b')$ and $\|c''\|_1 \leq 40k$.
		
	\end{proof}
	
	\section{Application of one-variable decoupling: Tail bounds for Boolean/Gaussian functions}
	\label{sec:tailbound}
	
	In this section, we will prove the following result using one-variable decoupling inequality.
	
	\begin{reptheorem}{thm:tail}
		There is a universal constant $C$ such that the following holds: Suppose $f:\{-1, 1\}^n \to \R$ has degree at most $k$ and that $\Var[f] = 1$. Let $t \geq 1$ and suppose that $\MaxInf[f] \leq C^{-k^3t^2}$. Then
		\[
		\Pr[|f|\geq t] \geq \exp(-Ct^2k^2)
		\]
	\end{reptheorem}
	
	With this better tail bound, repeat  the exact proof in Section 5 in \cite{DFKO}, we show a bounded real-valued function whose Fourier tail is very small must be close to a junta.
	
	\begin{reptheorem}{thm:junta}
		Let $f : \{-1, 1\}^n \to [-1, 1]$, $k\geq 1$, and $\epsilon > 0$. Suppose
		\[
		\sum_{|S| > k} \widehat{f}(S)^2 \leq \exp(-O(k^2)/\eps)
		\]
		Then $f$ is an $(\epsilon, 2^{O(k^3)}/\epsilon^2)$-junta.
	\end{reptheorem}
	
	\subsection{A lower bound on large deviations}
	
	We begin with a lower bound of tail probability for one-variable decoupled polynomials with Gaussian inputs.
	
	\begin{lemma}
		\label{lem:decoupledtail}
			There exists a universal constant $C$ such that the following holds:
			Let $f:\R^{n+m} \to \R$ with degree $k$ and $f$ is a linear map on a subset of inputs:
			\[
			f(y,z) = \sum_{i=1}^n y_i g_i(z).
			\]
			where $y = (y_1, \dots, y_n) \in \R^n$, $z \in \R^m$ and, $g_i:\R^m \to \R$ with degree at most $k - 1$ for all $i$. Suppose
			\[
			\Var[f]=\sum_{i}\sum_{S}\widehat{g_i}(S)^2 \geq 1.
			\] Then for Gaussian random variables $\bY \sim N(0, 1)^n, \bZ \sim N(0, 1)^m$,
			\[
			\Pr_{\bY, \bZ}[|f(\bY, \bZ)|\geq t] \geq \exp(-O(k + t^2)).
			\]
	\end{lemma}
	\begin{proof}
		By Parseval's Theorem,
		\[
		\E_{\bZ}\left[\sum_{i=1}^ng_i(\bZ)^2\right] = \sum_{i = 1}^n \E_{\bZ}[g_i(\bZ)^2] = \sum_{i = 1}^n \sum_{S}\widehat{g_i}(S)^2 = \Var[f] \geq 1
		\]
		We denote $l_{z}(y) = f(y, z)$ for any $z \in \R^m$.$\sum_{i=1}^ng_i(\bZ)^2$ is a polynomial with degree at most $2(k-1)$ and by Lemma~\ref{lem:exceedmean},
		\[
		\Pr_{\bZ}[\Var[l_{\bZ}] \geq 1] \geq \Pr_{\bZ}\left[\sum_{i=1}^ng_i(\bZ)^2 \geq \E_{\bZ}\left[\sum_{i=1}^ng_i(\bZ)^2\right]\right] \geq \exp(-O(k)).
		\]
		Notice that $l_z$ is a sum of independent Gaussian variables. Therefore,
		\[
		\Pr_{\bY} [l_{z}(\bY) \geq t] \geq \exp(-O(t^2))
		\]
		when $\Var[l_{z}] \geq 1$. Therefore
		\[
		\Pr_{\bY, \bZ}[|f(\bY, \bZ)|\geq t] \geq \Pr_{\bZ}[\Var[l_{\bZ}] \geq 1] \Pr_{\bY, \bZ} [l_{\bZ}(\bY) \geq t | \Var[l_{\bZ}] \geq 1] \geq \exp(-O(k+t^2))
		\]
	\end{proof}
	
	\begin{remark}
		In \cite{LT13}, they show that a linear function with Boolean inputs has a similar tail bound when the coefficients are sufficiently small. Therefore Theorem~\ref{lem:decoupledtail} also holds for Boolean inputs, with constraints that $\Inf_{y_i}(f) = \sum_{S}\widehat{g_i}(S)^2 \leq t^{-2}C^{-k}$ for all $i$, where $C$ is a universal constant. The proof is similar but slightly twisted by using Bonami inequality and the lemma in \cite{LT13}.
	\end{remark}
	
	\begin{proof}[Proof of Theorem~\ref{thm:tail}]
		Theorem~\ref{thm:one-variable} shows us an inequality between the tail probability of a polynomial $f$ and its one-variable decoupled version $\breve{f}$ with Gaussian inputs. Notice that $\Var[\breve{f}] \geq \Var[f]$ and $\breve{f}(y, z)$ is a linear map on $y$. Therefore if $\Var[f] = 1$, combining Theorem~\ref{thm:one-variable} and Lemma~\ref{lem:decoupledtail} we get
		\[
		\Pr_{\bX}[|f(\bX)| \geq t] \geq \frac1{O(k)} \Pr_{\bY, \bZ} [|\breve{f}(\bY, \bZ)| \geq Ckt] \geq \exp(-O(t^2k^2))
		\]
		The tail probabilities for Gaussian and Boolean inputs are close when all influences are small. By Lemma~\ref{lem:invariance}, there exists some universal constant $C$, such that
		\[
		\left|\Pr_{\bx \sim \{-1,1\}^n}[f(\bx) \geq t] - \Pr_{\bX \sim N(0,1)^n}[f(\bX) \geq t]\right| \leq \exp(-Ct^2k^2)
		\]
		when $\MaxInf[f] \leq C^{-t^2k^3}$. Theorem~\ref{thm:tail} holds when we can choose $C$ sufficienly large to let the difference of tail probabilities between Boolean and Gaussian inputs very tiny comparing to the tail bound of Gaussian function.
		\end{proof}
	
	\begin{remark}
		When $f$ is a homogenuous function with degree $k$, it is easy to check that $\Var[\breve{f}]= k\Var[f]$. Then following the proof of Theorem~\ref{thm:tail} we get a better tail bound (reduce $k$ in exponential).
		\begin{corollary}
			\label{cor:homo}
			There is a universal constant $C$ such that the following holds: Suppose $f:\{-1, 1\}^n \to \R$ is homogenuous and has degree at most $k$ and that $\Var[f] = 1$. Let $t \geq 1$ and suppose that $\MaxInf[f] \leq C^{-t^2k^2}$. Then
			\[
			\Pr[|f|\geq t] \geq \exp(-Ct^2k)
			\]
		\end{corollary}
	\end{remark}
	
	\subsection{Tightness of the tail bound}
	In \cite{DFKO}, they present some Boolean/Gaussian function matching the bound of Theorem \ref{thm:tail} and \ref{thm:junta}.
	In this section, we will show that Corollary \ref{cor:homo} is also tight when considering either $k$ or $t$ is constant.
	\begin{theorem}
		\label{thm:tightk}
		For any $k$ , there exist some degree-$k$ homogeneous polynomial $f: \R^n \to \R$ with $\Var[f] = 1$ such that
		\[
		\Pr_{\bX \sim N(0,1)^n}[|f(\bX)| \geq t] \leq \exp(-\Omega(kt^{2/k})).
		\]
	\end{theorem}
	
	\begin{proof}
		Consider about function $g(x) = x^k$ Then we have
		\[
		\Pr_{\bX_0 \sim N(0,1)}[|g(\bX_0)| \geq t] = \Pr_{\bX_0 \sim N(0,1)}[|\bX_0| \geq t^{1/k}] \leq \exp(-\Omega(t^{2/k})).
		\]
		Here we can rewrite variable $\bX_0 = \frac{1}{\sqrt{n}}(\bX_1 + \bX_2+ \dots \bX_n)$ where $\bX_i \sim N(0,1)$. We define function
		\[
		f_n(x) = \frac{1}{\sqrt{{n \choose k}}}\sum_{S \in {[n] \choose k}} \prod_{i\in S} x_i
		\]
		and we can use $f_n(x)$ to approximate normalized $f$. Therefore,
		\[
		\lim_{n \to \infty} \Pr_{\bX \in N(0,1)^n}[|g_n(\bX)|\geq t] = \Pr_{\bX \in N(0,1)^n}\left[\frac{1}{\sqrt{(k-1)!}}\left|f\left(\frac{\sum_{i=1}^n\bX_i}{\sqrt{n}}\right)\right|\geq t \right] \leq \exp(-\Omega(kt^{2/k}))
		\]
	\end{proof}
	
	\begin{theorem}
		\label{thm:tightk}
		For any constant $k$, for any $t > 1$, there exist some degree-$k$ homogeneous Gaussian function $f: \R^n \to \R$ with $\Var[f] = 1$ such that
		\[
		\Pr_{\bX \sim N(0,1)^n}[|f(\bX)| \geq t] \leq \exp(-\Omega(t^2))
		\]
	\end{theorem}
	\begin{proof}
		Consider about function
		\[
		f_n(x) = \frac{1}{\sqrt{n}}(x_1 x_2 \dots x_k+ x_{k+1}x_{k+2}\dots x_{2k}+\dots + x_{(n-1)k+1}x_{(n-1)k+2} \dots x_{nk})
		\]
		We know that if $\bX_1, \dots, \bX_k \sim N(0,1)$ are independent Gaussians,  $\E[\bX_1\bX_2\dots\bX_k] = 0, \E[(\bX_1\bX_2\dots\bX_k)^2] = 1$ and $\E[|\bX_1\bX_2\dots\bX_k|^3]= \rho^k$ where $\rho = \E_{\bX \sim N(0,1)}[|\bX|^3]$ is some constant not related to $k$.
		By Berry-Esseen Theorem,
		\[
		\left|\Pr_{\bX \sim N(0,1)^{kn}}[|f_n(\bX)| \geq t] - \Pr_{\bX_0 \sim N(0,1)}[|\bX_0| \geq t]
		\right| \leq \frac{C\rho^k}{\sqrt{n}}.
		\]
		If we choose a sufficient large $n \geq \rho^{2k} \exp(t^6) $, function $f_n$ will approximate standard Gaussian distribution. Therefore,
		\[
		\Pr_{\bX \sim N(0,1)^{kn}}[|f_n(\bX)| \geq t] \leq \Pr_{\bX_0 \sim N(0,1)}[|\bX_0| \geq t] + \frac{C\rho^k}{\sqrt{n}} \leq   \exp(-t^2) - C\exp(-t^3) \leq \exp(-\Omega(t^2))
		\]
	\end{proof}
	
	\section{Application of fully-Decoupling: Query complexity for fully-decoupled functions}
	\label{sec:application2}
	
	In \cite{aaronson2014}, they presented the following randomized algorithm which is a simulation of quantum algorithms.
	
	\begin{reptheorem}{thm:queryforblock}
		Let $f : \{-1,1\}^{n} \to [-1,1]$ be any block-multilinear polynomial of degree $k$ bounded in $[-1, 1]$. Then there exists a classical randomized algorithm that, on input $x \in \{-1,1\}^{n}$, non adaptively queries $O(\frac1{\eps^2}(2e)^{2k}n^{1-1/k})$ bits of $x$, and then outputs an estimate $f^*$ such that with high probability,
		\[
		|f^* - f(x)| \leq \epsilon
		\]
	\end{reptheorem}
	
	However they conjectured that the block-multilinearility condition can be removed. In this section, we will prove this conjecture via decoupling method.
	
	\begin{reptheorem}{thm:query}
		Suppose $f : \{-1,1\}^{n} \to [-1,1]$ has degree at most $k$. Then there exists a classical randomized algorithm that, on input $x \in \{-1,1\}^{n}$, non adaptively queries $O(\frac1{\eps^2}C^kn^{1-1/k})$ bits of $x$, where $C$ is a universal constant and then outputs an estimate $f^*$ such that with high probability,
		\[
		|f^* - f(x)| \leq \epsilon.
		\]
	\end{reptheorem}
	
	\begin{proof}[Proof of Theorem~\ref{thm:query} for non-block-multilinear case]
		Suppose function $f: \{-1,1\}^n \to [-1,1]$  with degree $k$. WOLOG we consider $f$ is homogeneous, otherwise we can add at most $k$ fake variables to make all terms be degree $k$. Notice in Theorem~\ref{thm:fulldecouple}, if we set convex function $\Phi(x) = |x|^m$, by letting $m \to \infty$, we can get a decoupling result for the maximum for homogeneous case:
		\[
		\max|\tilde{f}| \leq \exp(O(k)) \max |f|
		\]
		Therefore the output of $\tilde{f}$ is bounded in $[-C^k, C^k]$ for some constant $C$. Since $\tilde{f}/{C^k}$ is fully-decoupled(block-multilinear) and bounded in $[-1, 1]$, we can use the algorithm in Theorem~\ref{thm:queryforblock} to estimate $\tilde{f}/{C^k}$ using $O(\frac{1}{\epsilon^2} (2e)^{2k} n^{1-1/k})$ queries within error $\epsilon$. Another property for fully decoupling is that $\tilde{f}(x, x, \dots, x) = f(x)$. Therefore we can estimate $f$ in non-adaptive $O(\frac{1}{\epsilon^2}C'^kn^{1-1/k})$ queries within error $\epsilon$ with high probability where $C' = (2e)^2C$.
	\end{proof}

	\appendix
	
	\section{Bug fixing of the proof of Lemma 21 in \cite{aaronson2014}}
	\label{sec:appendix}
	
	There is a small bug in the proof of Lemma 21 in \cite{aaronson2014}. In the proof of Lemma 21, they denote
	\[
		V_i = \sum_{i_2, \dots, i_k \in [N]} a_{i_2, \dots, i_k}^2,
	\]
	and
	\[
		X_i = \sum_{i_2, \dots, i_k in [N]} a_{i_2, \dots, i_k} x_{i_2}\dots x_{i_k},
	\]
	where $x_{i_2}, \dots, x_{i_k}$ are uniformly random bits in $\{-1, 1\}$. Then $\E[X_i^2] = V_i$. They claim that ``By the concavity of the square root function, this means $\E[|X_i|] \geq \sqrt{V_i}$''. Here the direction of the inequality is incorrect. We can fix it by the following lemma in \cite{Ryan2014}.
	
	\begin{lemma}[Theorem 9.22 in \cite{Ryan2014}]
		\label{lemma:rootbound}
		Let $f : \{-1,1\}^n\to \R$ be a function  of degree at most $k$. Then
		\[
		\|f\|_2 \leq e^k \|f\|_1
		\]
	\end{lemma}
	
	Since $X_i$ is a function of degree $k-1$, by plugging in this lemma we get
	\[
		\sqrt{V_i} = \sqrt{\E[X_i^2]} \leq e^k \E[|X_i|]
	\]
	Therefore $\sqrt{V_i}$ is still upper bounded by $\E[|X_i|]$, but we introduce a factor of $e^k$. The rest of the proof still holds and we need to add a factor of $e^k$ in the query complexity of the algorithm. But since the query complexity of the original theorem has already contained a factor exponential in $k$, this bug does not make a big difference.
}	
		
	\bibliographystyle{alpha}
	\bibliography{yuzhao}

\newcommand{\etalchar}[1]{$^{#1}$}
\begin{thebibliography}{BMO{\etalchar{+}}15}

\bibitem[AA14]{AA14}
Scott Aaronson and Andris Ambainis.
\newblock The need for structure in quantum speedups.
\newblock {\em Theory Of Computing}, 10(6):133--166, 2014.

\bibitem[AA15]{AA15}
Scott Aaronson and Andris Ambainis.
\newblock Forrelation: a problem that optimally separates quantum from
  classical computing.
\newblock In {\em Proceedings of the 47th Annual ACM Symposium on Theory of
  Computing}, pages 307--316, 2015.

\bibitem[Aar05]{Aar05a}
Scott Aaronson.
\newblock Ten semi-grand challenges for quantum computing theory, 2005.
\newblock \url{http://www.scottaaronson.com/writings/qchallenge.html}.

\bibitem[Aar08]{Aar08}
Scott Aaronson.
\newblock How to solve longstanding open problems in quantum computing using
  only {F}ourier {A}nalysis.
\newblock Lecture at Banff International Research Station, 2008.
\newblock \url{http://www.scottaaronson.com/talks/openqc.ppt}.

\bibitem[Aar10]{Aar10a}
Scott Aaronson.
\newblock Updated version of ``ten semi-grand challenges for quantum computing
  theory'', 2010.
\newblock \url{http://www.scottaaronson.com/blog/?p=471}.

\bibitem[BMO{\etalchar{+}}15]{BMO+15}
Boaz Barak, Ankur Moitra, Ryan O'Donnell, Prasad Raghavendra, Oded Regev, David
  Steurer, Luca Trevisan, Aravindan Vijayaraghavan, David Witmer, and John
  Wright.
\newblock Beating the random assignment on constraint satisfaction problems of
  bounded degree.
\newblock In {\em Proceedings of the 18th Annual International Workshop on
  Approximation Algorithms for Combinatorial Optimization Problems}, 2015.

\bibitem[Bou02]{Bou02}
Jean Bourgain.
\newblock On the distribution of the {F}ourier spectrum of {B}oolean functions.
\newblock {\em Israel Journal of Mathematics}, 131(1):269--276, 2002.

\bibitem[DFKO07]{DFKO07}
Irit Dinur, Ehud Friedgut, Guy Kindler, and Ryan O'Donnell.
\newblock On the {F}ourier tails of bounded functions over the discrete cube.
\newblock {\em Israel Journal of Mathematics}, 160(1):389--412, 2007.

\bibitem[Din07]{Din07}
Irit Dinur.
\newblock The {PCP} {T}heorem by gap amplification.
\newblock {\em Journal of the ACM}, 54(3):1--44, 2007.

\bibitem[dlP92]{dlPen92}
Victor de~la Pe{\~n}a.
\newblock Decoupling and {K}hintchine's inequalities for {$U$}-statistics.
\newblock {\em Annals of Probability}, 20(4):1877--1892, 1992.

\bibitem[dlPG99]{dlPG99}
V\'{i}ctor de~la Pe{\~n}a and Evarist Gin\'{e}.
\newblock {\em Decoupling: from dependence to independence}.
\newblock Springer, 1999.

\bibitem[dlPMS95]{dlPM95}
Victor de~la Pe{\~n}a and Stephen Montgomery-Smith.
\newblock Decoupling inequalities for the tail probabilities of multivariate
  {$U$}-statistics.
\newblock {\em Annals of Probability}, 23(2):806--816, 1995.

\bibitem[EFF12]{EFF12}
David Ellis, Yuval Filmus, and Ehud Friedgut.
\newblock Triangle-intersecting families of graphs.
\newblock {\em Journal of the European Mathematical Society}, 14(3):841--885,
  2012.

\bibitem[FK96]{FK96}
Ehud Friedgut and Gil Kalai.
\newblock Every monotone graph property has a sharp threshold.
\newblock {\em Proceedings of the American Mathematical Society},
  124(10):2993--3002, 1996.

\bibitem[FKN02]{FKN02}
Ehud Friedgut, Gil Kalai, and Assaf Naor.
\newblock Boolean functions whose {F}ourier transform is concentrated on the
  first two levels and neutral social choice.
\newblock {\em Advances in Applied Mathematics}, 29(3):427--437, 2002.

\bibitem[Fri98]{Fri98}
Ehud Friedgut.
\newblock Boolean functions with low average sensitivity depend on few
  coordinates.
\newblock {\em Combinatorica}, 18(1):27--35, 1998.

\bibitem[Gin98]{Gin98}
Evarist Gin\'e.
\newblock A consequence for random polynomials of a result of de la {P}e\~{n}a
  and {M}ontgomery-{S}mith.
\newblock In {\em Probability in Banach Spaces 10}, volume~43 of {\em Progress
  in Probability}. Birkh\"auser--Verlag, 1998.

\bibitem[JOW12]{JOW12}
Jacek Jendrej, Krzysztof Oleszkiewicz, and Jakub Wojtaszczyk.
\newblock On some extensions of the {FKN} theorem.
\newblock Manuscript, 2012.
\newblock To appear in \emph{Theory of Computation}.

\bibitem[Kan11]{Kan11}
Daniel Kane.
\newblock $k$-independent {G}aussians fool polynomial threshold functions.
\newblock In {\em Proceedings of the 26th Annual Computational Complexity
  Conference}, pages 252--261, 2011.

\bibitem[Kho02]{Kho02}
Subhash Khot.
\newblock On the power of unique 2-prover 1-round games.
\newblock In {\em Proceedings of the 34th Annual ACM Symposium on Theory of
  Computing}, pages 767--775, 2002.

\bibitem[Kin02]{Kin02}
Guy Kindler.
\newblock {\em Property Testing, {PCP}, and juntas}.
\newblock PhD thesis, Tel Aviv University, 2002.

\bibitem[KM13]{KM13}
Daniel Kane and Raghu Meka.
\newblock A {PRG} for {L}ipschitz functions of polynomials with applications to
  {S}parsest {C}ut.
\newblock In {\em Proceedings of the 45th Annual ACM Symposium on Theory of
  Computing}, pages 1--10, 2013.

\bibitem[KN06]{KN06}
Subhash Khot and Assaf Naor.
\newblock Nonembeddability theorems via {F}ourier analysis.
\newblock {\em Mathematische Annalen}, 334(4):821--852, 2006.

\bibitem[KN08]{KN08}
Subhash Khot and Assaf Naor.
\newblock Linear equations modulo 2 and the {$L_1$} diameter of convex bodies.
\newblock {\em SIAM Journal on Computing}, 38(4):1448--1463, 2008.

\bibitem[KO12]{KO12}
Guy Kindler and Ryan O'Donnell.
\newblock Gaussian noise sensitivity and {F}ourier tails.
\newblock In {\em Proceedings of the 27th Annual Computational Complexity
  Conference}, pages 137--147, 2012.

\bibitem[KS02]{KS02}
Guy Kindler and Shmuel Safra.
\newblock Noise-resistant {B}oolean functions are juntas.
\newblock Manuscript, 2002.

\bibitem[Kwa87]{Kwa87}
Stanis{\l}aw Kwapie\'n.
\newblock Decoupling inequalities for polynomial chaos.
\newblock {\em Annals of Probability}, 15(3):1062--1071, 1987.

\bibitem[Lov10]{Lov10}
Shachar Lovett.
\newblock An elementary proof of anti-concentration of polynomials in
  {G}aussian variables.
\newblock Technical Report 182, Electronic Colloquium on Computational
  Complexity, 2010.

\bibitem[MOO10]{MOO10}
Elchanan Mossel, Ryan O'Donnell, and Krzysztof Oleszkiewicz.
\newblock Noise stability of functions with low influences: invariance and
  optimality.
\newblock {\em Annals of Mathematics}, 171(1):295--341, 2010.

\bibitem[MS58]{MS58}
Nathaniel Macon and Abraham Spitzbart.
\newblock Inverses of {V}andermonde matrices.
\newblock {\em The American Mathematical Monthly}, 65:95--100, 1958.

\bibitem[MS14]{MS14}
Konstantin Makarychev and Maxim Sviridenko.
\newblock Solving optimization problems with diseconomies of scale via
  decoupling.
\newblock In {\em Proceedings of the 55th Annual IEEE Symposium on Foundations
  of Computer Science}, pages 571--580, 2014.

\bibitem[O'D14]{OD14}
Ryan O'Donnell.
\newblock {\em Analysis of Boolean Functions}.
\newblock Cambridge University Press, 2014.

\end{thebibliography}
	
\end{document}